\documentclass[prodmode,acmec]{ec-acmsmall} 

\usepackage{amsmath}
\usepackage{amssymb}
\usepackage{balance}
\usepackage{graphicx, subfigure}
\usepackage{times,color}
\usepackage{ifthen}
\usepackage[ruled]{algorithm2e}
\renewcommand{\algorithmcfname}{ALGORITHM}
\SetAlFnt{\small}
\SetAlCapFnt{\small}
\SetAlCapNameFnt{\small}
\SetAlCapHSkip{0pt}
\IncMargin{-\parindent}

\acmVolume{X}
\acmNumber{X}
\acmArticle{X}
\acmYear{2015}
\acmMonth{2}
\usepackage[numbers]{natbib}

\newtheorem{claim}[theorem]{Claim}

\newtheorem{invariant}[theorem]{Invariant}

\newenvironment{proofof}[1]{{\sc Proof of #1:} }
{\qed\par\vskip 4mm\par}

\newcommand{\class}{\mathcal{C}}
\newcommand{\CompetitiveRatio}{\textrm{cr}}
\newcommand{\ffrac}[2]{#1/#2}
\newcommand{\IntegralityGap}{\textrm{IG}}

\newcommand{\OPT}{OPT}

\newcommand{\vd}{\rho}

\newcommand{\y}[3]{
  \ifthenelse { \equal{#2}{} }
     {y_{#1}(#3)}
     {
        \ifthenelse{ \equal{#3}{}}
            {y_{#1}^{#2}}
            {y_{#1}^{#2}(#3)}
     }
}

\newcommand{\JobInput}{\mathcal{J}}

\newcommand{\TimeInput}{\mathcal{T}}
\newcommand{\TimeFull}{\TimeInput^{F}}
\newcommand{\TimePartial}{\TimeInput^{P}}

\newcommand{\admitted}[1]{#1^{(a)}}

\newcommand{\virtual}[1]{#1^{(v)}}
\newcommand{\type}{\tau}

\newcommand{\Alg}{\mathcal{A}}

\newcommand{\gap}{\mu}
\newcommand{\thres}{\gamma}

\begin{document}

\markboth{Y. Azar et al.}{Truthful Online Scheduling with Commitments}

\title{Truthful Online Scheduling with Commitments}
\author{Yossi Azar
\affil{Blavatnik School of CS, Tel Aviv University, Tel Aviv, Israel}
Inna Kalp-Shaltiel
\affil{Blavatnik School of CS, Tel Aviv University, Tel Aviv, Israel}
Brendan Lucier
\affil{Microsoft Research, Cambridge, MA}
Ishai Menache
\affil{Microsoft Research, Redmond, WA}
Joseph (Seffi) Naor
\affil{CS Department, Technion, Haifa, Israel}
Jonathan Yaniv
\affil{CS Department, Technion, Haifa, Israel}
}


\begin{abstract}
We study online mechanisms for preemptive scheduling with deadlines, with the goal of maximizing the total value of completed jobs.  This problem is fundamental to deadline-aware cloud scheduling, but there are strong lower bounds even for the algorithmic problem without incentive constraints.  However, these lower bounds can be circumvented under the natural assumption of deadline slackness, i.e., that there is a guaranteed lower bound $s > 1$ on the ratio between a job's size and the time window in which it can be executed.

In this paper, we construct a truthful scheduling mechanism with a constant competitive ratio, given slackness $s > 1$.  Furthermore, we show that if 
$s$ is large enough then we can construct a mechanism that also satisfies a \emph{commitment} property: it can be determined whether or not a job will finish, and the requisite payment if so, well in advance of each job's deadline.
This is notable because, in practice, users with strict deadlines may find it unacceptable to discover only very close to their deadline that their job has been rejected.

\end{abstract}


%
%
%


\begin{bottomstuff}
This work is supported in part by the Technion-Microsoft Electronic Commerce Research Center, by the Israel Science Foundation (grant No. 1404/10) and by the Israeli Centers of Research Excellence (I-CORE) program (Center No. 4/11).
%
%
\end{bottomstuff}

\maketitle


\section{Introduction}


Modern computing applications, such as search engines and big-data processing, run
on large clusters operated by either first or third parties (a.k.a., private and public clouds, respectively).
Since end-users do not own the compute infrastructure, the use of cloud computation necessitates crisp contracts between them and the cloud provider on the service terms (i.e., Service Level Agreements - SLAs).
The problem of designing and implementing such contracts falls within the scope of online mechanism design, which concerns the design of mechanisms for allocating resources when agents arrive and depart over time, and the mechanism must make allocation decisions online.
A contract can be as simple as renting out a virtual machine for a certain price per hour.  However, with the increased variety of cloud-offered services come more performance-centric contracts, such as paying per number of transactions \cite{AzureMLPricing}, or a guarantee to finish executing a job by a certain deadline \cite{rayon,jockey}.

Since the underlying physical resources are often limited,
a cloud provider faces resource management challenges, such as deciding which service requests to accept in view of the required SLAs, and determining how best to schedule or allocate resources to the different users.
%
For instance, the provider may opt to delay time-insensitive tasks when usage peaks, or prevent admission of low-priority jobs if higher-priority jobs are expected to arrive.  To make these decisions in a principled manner, one wishes to design a mechanism for an online scheduling problem with deadlines, aimed at maximizing the total value of completed jobs.  This social welfare objective is particularly relevant in the private cloud setting. It is also relevant for markets with competition between cloud providers, where each provider wishes to extend its market share by increasing user satisfaction.
At a high level, the goal of this paper is to provide algorithmic foundations for scheduling jobs with different demands, values and deadlines, in a manner that is compatible with user incentives.

The problem can be abstracted as follows.  Each job request $j$ is associated with an arrival time $a_j$, a size (demand) $D_j$, a deadline $d_j$ and a value $v_j$.  There are $C$ identical machines that can process jobs. Each job uses at most a single machine at a time, and jobs can be preempted and resumed.  The goal is to maximize the total value of jobs completed by their deadlines. In a perfect world, a solution to this problem would achieve a good competitive ratio, would be incentive compatible, and would notify jobs whether or not they are completed as swiftly as possible.
Unfortunately, the basic online scheduling problem, without considering incentives or commitments, is inherently difficult even when $C=1$. From a worst-case perspective, there is a polylogarithmic lower bound on the competitive ratio of any randomized algorithm \cite{CI98}.
However, the known lower bounds only apply in the presence of jobs with tight deadlines (i.e., $d_j = a_j + D_j$). Recent work circumvented the lower bound by assuming \emph{deadline slackness}, where every job $j$ satisfies $d_j - a_j \geq s \cdot D_j$ for a slackness parameter $s > 1$ \cite{LMNY13}.
Our aim is to continue this line of inquiry and design incentive compatible scheduling mechanisms in the presence of deadline slackness.

\vspace{0.03in}
\paragraph{Truthfulness}
In our online scheduling context, the incentive compatibility requirement is multi-parameter: agents must be incentivized to report their tuple of job parameters $\langle v_j, D_j, a_j, d_j \rangle$. As is standard, we assume agents cannot deviate to an arrival time earlier than $a_j$, nor report a deadline later than $d_j$.
These assumptions are natural if one views the arrival time as the first time the customer is able to interact with the mechanism, and that job results are not released to a customer until the reported deadline.  Furthermore, we generally assume that a job holds no value to the customer unless it is fully completed. Hence, a user cannot benefit from underreporting the job demand.

\vspace{0.03in}
\paragraph{Commitments}
In addition to incentive compatibility, another important feature of a practical scheduling mechanism is commitment: whether, and when, a scheduler guarantees to complete a given job.  Traditionally, a preemptive scheduler is allowed to accept a job, process it partially, but then abandon it once its deadline has passed.  While this behavior may be justified in terms of pure optimization, in many real-life scenarios it is not acceptable, since users might be left empty-handed at their deadline.  In reality, users with business-critical jobs require an indication, well before their deadline, of whether their jobs can be processed.  Since sustaining deadlines is becoming a key requirement for modern computation clusters (e.g., \cite{rayon} and references therein), it is essential that schedulers provide some degree of commitment.  

The question is: at what point of time should the scheduler commit to jobs? One option is to require the scheduler to commit to jobs upon arrival. Namely, once a job arrives, the scheduler immediately decides whether it accepts the job (and then it is required to complete it) or reject the job. However, \cite{LMNY13} proved that for general values no scheduler can commit to jobs upon arrival while providing any performance guarantees, even assuming deadline slackness.
Therefore, a more plausible alternative from the user perspective is to allow the committed scheduler to delay the decision, but only up to some predetermined point.

\vspace{0.03in}
\begin{definition}
\label{def:betaResponsive}
A scheduling mechanism is called \emph{$\beta$-responsive} (for $\beta \geq 0$) if, for every job $j$, by time $d_j - \beta \cdot D_j$ it either (a) rejects the job, or (b) guarantees that the job will be completed by its deadline and specifies the required payment.
\end{definition}

\noindent
Note that $\beta$-responsiveness requires deadline slackness $s \geq \beta$ for feasibility. Schedulers that do not provide advance commitment are by default $0$-responsive; we often refer to them as being non-committed.
Useful levels of commitments are typically obtained when $\beta \ge 1$, as this provides rejected users an opportunity to execute their job elsewhere before their deadline.

One might consider different definitions for responsiveness in online scheduling. In a sense, the definition given here is additive: for each job $j$, the mechanism must make its decision $\beta D_j$ time units before the deadline. An alternative definition could be fractional: the decision must be made before some fraction of job execution window, e.g., $d_j - \omega (d_j - a_j)$ for $\omega \in (0,1)$. It turns out that many of our results\footnote{Specifically, all of the results stated in Section \ref{sec:Introduction_OurResults}, except for Theorem \ref{thm.general}.} also satisfy responsiveness under this alternative definition, as well as other useful properties\footnote{Such as the \emph{no-early-processing} property: the scheduler cannot begin to process a job without committing first to its completion. This implies that any job that begins processing is guaranteed to complete.
}.
We discuss this further in Section \ref{sec:Conclusions}.

\subsection{Our Results}
\label{sec:Introduction_OurResults}
We design the first truthful online mechanisms for preemptive scheduling with deadlines. Moreover, our mechanism can be made $\beta$-responsive as defined above.

\medskip

\noindent
\textbf{Main Theorem (informal):} For every $\beta \geq 0$, given sufficiently large slackness $s \geq s(\beta)$, there is a truthful, $\beta$-responsive, $O(1)$-competitive mechanism for online preemptive scheduling on $C$ identical servers.

\medskip

\noindent
The precise competitive ratio achieved by our mechanism depends on the level of input slackness. We establish the main result in two steps.  First, we build a mechanism that is truthful, but not committed. Second, we develop a reduction from the problem of scheduling with responsive commitment to the problem of scheduling without commitment.
Each of these two steps may be of interest in their own right. In particular, we obtain in the first step a truthful $O(1)$-competitive mechanism for online preemptive scheduling with deadlines.


\begin{theorem}
There is a truthful mechanism for online scheduling on multiple identical servers that obtains a competitive ratio of
$2 + \Theta\Big(\frac{1}{\sqrt[3]{s}-1}\Big) + \Theta\Big(\frac{1}{(\sqrt[3]{s}-1)^3}\Big)$
for any $s>1$.
\end{theorem}

Note that, as implied by known lower bounds, this competitive ratio grows without bound as $s \to 1$.  However, as $s$ grows large, the competitive ratio we achieve approaches $2$.  Our approach for this result is to begin with a greedy scheduling rule that prioritizes jobs by value density (value per size), then modify this scheduler so that (a) jobs are not allowed to begin executing too close to their deadlines, and (b) one job cannot preempt another unless its value density is sufficiently greater.  These modifications generate incentive issues that need to be addressed with some additional tweaking.  We then analyze the competitive ratio of this scheduler using dual fitting techniques, as described in Section \ref{sec:DualFitting}. This analysis appears in Section \ref{sec:NonCommitted_Truthfulness}.



For the second step, we provide a general reduction from committed scheduler design to non-committed scheduler design. We will describe reduction here for $\beta=\ffrac{s}{2}$. The idea behind the reduction is to employ simulation: each incoming job is slightly modified and submitted to a simulator for the first half of its execution window. The simulator uses the given non-committed scheduling to ``virtually" process jobs. If the simulation completes a job, then the algorithm commits to executing the job on the physical server. See Section \ref{sec:Committed} for more details. This reduction can be applied to any scheduling algorithm, not just the truthful scheduler described above. Specifically, applying our reduction to the (non-truthful) algorithm described in \cite{LMNY13} generates a (non-truthful) committed scheduler with a competitive ratio that approaches $5$ as $s$ grows large.


\begin{theorem}
\label{thm.alg.responsive}
There is a $(\ffrac{s}{2})$-responsive scheduler for online scheduling on multiple identical servers that obtains a competitive ratio of
$5 + \Theta\Big(\frac{1}{\sqrt[3]{s/4}-1}\Big) + \Theta\Big(\frac{1}{(\sqrt[3]{s/4}-1)^2}\Big)$
for any $s > 4$.
\end{theorem}


To obtain both truthfulness and responsiveness, we wish to compose our reduction with the truthful non-committed mechanism described above.  One challenge is that our basic reduction preserves truthfulness with respect to all parameters \emph{except} arrival time. We can therefore immediately obtain a constant competitive-ratio scheduling mechanism which is $(\ffrac{s}{2})$-responsive, given sufficient slackness; and truthful, given that jobs do not purposely delay their arrivals. For the single server case, we obtain the same asymptotic bound as in Theorem \ref{thm.alg.responsive} for $s>4$; see Section \ref{sec:Committed_Truthful}.


To yield our most general result, we explicitly construct a scheduling mechanism that obtains full truthfulness based on the truthful non-committed scheduler and a general reduction from committed scheduling to non-committed scheduling. The construction is rather technical and significantly increases the competitive ratio.
We obtain the following result, with constants $s_0 = 12$ and $c_0 = 9$ for the single-server case, and $s_0 = 139.872$ and $c_0 = 94.248$ for the case of multiple identical servers.


\begin{theorem}
\label{thm.general}
There exist constants $c_0$ and $s_0$ such that there is a truthful, $(\ffrac{2s}{s_0})$-responsive mechanism for online scheduling on multiple identical servers that obtains a competitive ratio of
$c_0 + \Theta\Big(\frac{1}{\sqrt[3]{\ffrac{s}{s_0}}-1}\Big) + \Theta\Big(\frac{1}{(\sqrt[3]{\ffrac{s}{s_0}}-1)^3}\Big)$
for any $s > s_0$.
\end{theorem}

\subsection{Related Work}
\label{subsec:literature}

%
%
Online preemptive scheduling models have been widely studied in the scheduling theory for various objectives, with value maximization results being of most relevance to our work.
%
%
Canetti and Irani \cite{CI98} consider the case of tight deadlines, obtaining a deterministic  lower bound of $\kappa$ and a randomized $\Omega\big(\sqrt{\ffrac{\log \kappa}{\log\log\kappa}}\big)$ lower bound, where $\kappa$ is the max-min ratio between either job values or job demands. Several upper bounds have been constructed \cite{KS92,KS94,CI98,Porter04}, with the best being a randomized $O(\log\kappa)$ algorithm.
%
%
In \cite{LMNY13}, we show that by incorporating a deadline slackness constraint, a non-committed online preemptive scheduler for the general value model exists, and prove a bound\footnote{The bound presented by \cite{LMNY13} can be generalized to this form.} of
$2 +
\Theta\big(\frac{1}{\sqrt[3]{s} - 1}\big) + \Theta\big(\frac{1}{(\sqrt[3]{s} - 1)^2}\big)
$ on its competitive ratio, which is constant for every $s>1$. However, \cite{LMNY13} do not provide any algorithmic guarantees for committed scheduling models.
%
%
Other constant competitive schedulers have  been known only for special cases. When all demands are identical, a $5$-competitive scheduler exists, which can be improved to $2$ assuming a discrete timeline \cite{HKMP05}. Another studied model is where the value of each job equals its demand; this model is known as the busy time maximization problem \cite{DP00,GNYZ02,BCKMS99} . These works can be combined to obtain a $1$-responsive algorithm with a competitive ratio of $\min\{5.83, 1 + \ffrac{1}{s}\}$; however, the algorithm cannot be extended to incorporate general values.

Much less is known about \emph{truthful} online scheduling mechanisms. Previous works (e.g., \cite{LaviSwamy07,ArcherTardos01}) focus mostly on offline settings with makespan as main objective. \cite{JMNY11,JMNY12} design incentive compatible algorithms for jobs with deadlines, but restrict attention to the offline setting.  Works on online truthful scheduling have largely focused on achieving the (non-constant) bounds from the algorithmic literature \cite{Porter04,HKMP05}. Finally, \cite{LMNY13} proposes a heuristic that is incentive compatible and $1$-responsive, but no formal bounds are provided for the competitive ratio of that heuristic.

\section{Preliminaries} \label{sec:preliminaries}
In this section we present the scheduling model and necessary definitions (Sections \ref{sec:Preliminaries_Model} and \ref{sec:Preliminaries_Trutfulness}). We then provide a brief overview of the dual fitting technique, which is used to analyze the proposed mechanisms (Section \ref{sec:DualFitting}).

\subsection{Scheduling Model}
\label{sec:Preliminaries_Model}

We consider a system consisting of $C$ identical servers, which are always available throughout time. The scheduler receives job requests over time. Denote by $\JobInput$ the set of all job requests received by the scheduler.
Each job request $j\in\JobInput$ is associated with a \emph{type} $\type_j = \left\langle v_{j},D_{j},a_{j},d_{j}\right\rangle$.
The type of each job $j$ consists of the job value $v_j$, the job resource demand (size) $D_j$, the arrival time $a_j$ and the deadline $d_j$. Write $T$ as the space of possible types. We denote by $\vd_j = \ffrac{v_j}{D_j}$ the value-density of job $j$.
The job requests in $\JobInput$ are revealed to the scheduler only upon arrival. The scheduler can allocate resources to jobs, provided that at any point each job is processed on at most one server and each server is processing at most one job. Preemption is allowed. Specifically, jobs may be paused and resumed from the point they were preempted. If a job is allocated to servers for a total time of $D_j$ during the interval $[a_j,d_j]$, then it is completed by the scheduler.

An instance of the scheduling problem is represented by a type profile $\type = \{\type_j : j\in\JobInput\}$.
Given a scheduling algorithm $\Alg$, denote by $\Alg(\type)$ the jobs that are fully completed by $\Alg$ on an instance $\type$, and by $v(\Alg(\type))$ their aggregate value. The goal of the scheduler is to maximize $v(\Alg(\type))$.
Let $\OPT$ denote the optimal offline algorithm. The quality of an online scheduler is measured by its competitive ratio, which is the worst case ratio between the optimal offline value and the value gained by the algorithm.
In this paper, we define the competitive ratio as a function of the input \emph{slackness}, defined $s \triangleq s(\type) = \min\big\{\frac{d_j - a_j}{D_j} \mid \type_j = \left\langle v_{j},D_{j},a_{j},d_{j}\right\rangle \in \type \big\}$.
The competitive ratio of an online algorithm $\Alg$ on inputs with slackness $s$, denoted $\CompetitiveRatio_\Alg(s)$, is given by:
\begin{equation}
\CompetitiveRatio_\Alg(s) = \max_{\type : s(\type)=s} \left\{ \frac{v(\OPT(\type))}{v(\Alg(\type))} \right\} \,\, \in \,\, [1,\infty).
\end{equation}

The following definitions refer to the execution of an online allocation algorithm $\Alg$ over an instance $\type$. We drop $\Alg$ and $\type$ from notation when they are clear from context. Time is represented by a continuous variable $t$.
For a scheduling algorithm $\Alg$, denote by $j_{\Alg}^i(t)$ the job running on server $i$ at time $t$ and by $\vd_{\Alg}^i(t)$ its value-density. We use $\y{j}{i}{t}$ as a binary\footnote{In Section \ref{sec:DualFitting} we extend the range of values $y_j^i(t)$ may receive. However, we will always treat it as an allocation indicator.}
variable indicating whether job $j$ is running on server $i$ at time $t$, i.e., whether $j=j_{\Alg}^i(t)$ or not. We often refer to the function $y_j^i$ as the \emph{allocation} of job $j$ on server $i$, and to $y_j$ as the \emph{allocation} of job $j$.

\subsection{Mechanisms and Incentives} \label{sec:Preliminaries_Trutfulness}

Each job in $\JobInput$ is owned by a rational agent (i.e., user), who submits it to the scheduling mechanism.
We will be studying direct revelation mechanisms, where each user participates by announcing its type $\type_j = \langle v_j, D_j, a_j, d_j \rangle$ from the space $T$ of possible types. A mechanism then consists of an allocation rule $\Alg : T^\JobInput \to \{0,1\}^\JobInput$ and a payment rule $p : T^\JobInput \to \mathbb{R}^\JobInput$. Writing $\Alg(\type)$ as the profile of allocations returned by the mechanism given type profile $\type$, we interpret $\Alg_j(\type)$ as an indicator for whether the job of customer $j$ is fully completed by its deadline. In general mechanisms can be randomized, in which case we can interpret $\Alg_j(\type) \in [0,1]$ as the expected allocation of agent $j$. However, all of the mechanisms we consider in this paper are deterministic. We will restrict our attention to online mechanisms, which are constrained to make scheduling decisions at each point in time without knowledge of jobs that arrive at future times. Agents have quasilinear utilities: given allocations $x$ and payments $p$, the utility of user $j$ is given by $u_j(\type) = v_j \Alg_j(\type) - p_j(\type)$.  

We adopt a model in which we only allow late reports of arrivals, early reports of deadlines, and increased reports of job lengths.  As discussed in the introduction, this assumption is justifiable in the context of allocating cloud resources.  We say a mechanism is \emph{truthful} if, subject to these restrictions on type reports, each user $j$ maximizes expected utility by reporting his true type to the mechanism, for any possible declarations of the other agents.

We will make heavy use of a characterization of truthfulness made by \cite{HKMP05}.  We say that a type $\type_j = \langle v_j, D_j, a_j, d_j \rangle$ \emph{dominates} $\type'_j = \langle v_j', D_j', a_j', d_j' \rangle$ if $v_j \geq v'_j$, $D_j \leq D'_j$, $a_j \geq a'_j$, and $d_j \leq d'_j$.  We then say that an algorithm $\Alg$ is \emph{monotone} if for any type profile $\mathbf{\tau}$, any $j$, and any $\type_j'$ that dominates $\type_j$, we have that $\Alg_j(\type_j, \mathbf{\type}_{-j}) \leq \Alg_j(\type'_j, \mathbf{\type}_{-j})$.  For deterministic algorithms, this means that if job $j$ is allocated under input profile $\mathbf{\type}$, then it will also be allocated if customer $i$'s report changes from $\type_j$ to a type that dominates $\type_j$.

\begin{theorem}[\cite{HKMP05}]
\label{thm:truth}
Given an allocation algorithm $\Alg$, there exists a payment rule $p$ such that mechanism $(\Alg, p)$ is truthful if and only if $\Alg$ is monotone.
\end{theorem}

\subsection{LP and Dual Fitting}
\label{sec:DualFitting}

Our competitive ratio analysis relies on a relaxed formulation of the problem as a linear program (LP). The relaxed LP formulation was suggested in \cite{JMNY11} and considered later in \cite{JMNY12,LMNY13}. In this paper, we do not require the LP formulation itself, but do rely on its dual. 
For completeness, we present below both the primal and dual programs. The primal program holds a variable $y_j^i(t)$ representing the allocation of a job $j\in\JobInput$ on server $i$ at time $t\in [a_j,d_j]$.
\\

\noindent \textbf{Primal Program.}
    \begin{alignat}{5}
          \max
          \label{eq:PrimalObjective} & \quad \,\,\,\,\,\,
          \sum_{j\in\JobInput} \, \sum_{i=1}^{C} \intop_{a_j}^{d_j} \, \vd_j y_j^i(t)dt & \\
          \label{eq:PrimalDemand} & \quad
          \,\,\,\,\,\,\,\sum_{i=1}^{C} \intop_{a_j}^{d_j} \,  y_j^i(t)dt \,\,\le\,\, D_{j} & \,\,\,\,\, & \quad \forall j \\
          \label{eq:PrimalCapacity} & \quad \sum_{j:t \in [a_j,d_j]} y_j^i(t) \,\,\le\,\, 1 & \,\,\,\,\, & \quad \forall i,t \\
          \label{eq:PrimalGapDecreasing} & \quad
          \,\,\,\,\,\,\,\sum_{i=1}^{C} y_j^i(t) - \frac{1}{D_j} \cdot  \sum_{i=1}^{C} \intop_{a_j}^{d_j} y_j^i(t)dt \,\,\le\,\, 0 & \,\,\,\,\, & \quad \forall j,t\in [a_j,d_j] \\
          & \quad \,\,\,\,\,\,\,\,\,\,y_j^i(t) \ge 0 & \,\,\,\,\, & \quad \forall j,i,t \in [a_j,d_j] \nonumber
    \end{alignat}

\noindent
The first two sets of constraints \eqref{eq:PrimalDemand},\eqref{eq:PrimalCapacity} are standard demand and capacity constraints. The constraints \eqref{eq:PrimalGapDecreasing} are gap-reducing constraints; see \cite{JMNY11} for an interpretation of these constraints.
Note that for the single server case, the constraints \eqref{eq:PrimalGapDecreasing} are redundant, since they follow from \eqref{eq:PrimalCapacity}.
The primal objective \eqref{eq:PrimalObjective} is to maximize the total (fractional) value.

The dual linear program of an instance $\type$ is given as follows.
\\

%
%
\noindent \textbf{Dual Program.}
    \begin{alignat}{5}
          \min        & \quad \sum_{j\in\JobInput} D_j \alpha_j  \,+\,  \sum_{i=1}^{C} \, \intop_{0}^{\infty} \beta_i(t) dt & \label{eq:DualObjective} & \\
          \text{s.t.} & \quad \,\,\,\alpha_j + \beta_i(t) + \pi_j(t) - \frac{1}{D_j}\intop_{a_j}^{d_j} \pi_j(t')dt' \,\,\ge\,\, \vd_j & & \quad \forall j\in\JobInput ,\, i ,\, t\in [a_j,d_j] \label{eq:DualConstraint}\\
          & \quad \,\,\,\alpha_j ,\, \beta_i(t) ,\, \pi_j(t) \,\,\ge\,\, 0 & \,\,\,\,\, & \quad \forall j\in\JobInput ,\, i ,\, t\in [a_j,d_j]
    \end{alignat}

\noindent We provide the intuition behind the dual formulation.
%
The dual program holds a constraint \eqref{eq:DualConstraint} for every tuple $(j,i,t)$, where $j$ is an input job, $i$ is a server index, and $t\in [a_j,d_j]$ is a specific time. Note that since time is continuous, there are an infinite number of constraints. However, this does not impose an issue, since we do not solve the dual program explicitly.
There are three types of dual variables. We typically set $\pi_j(t)=0$, since these variables are not required throughout this paper. The second variable $\alpha_j$ is associated with each job $j$ and appears in all of the constraints of job $j$.
Setting $\alpha_j = \vd_j$ allows us to satisfy all of the constraints associated with job $j$.
As a result, the dual objective function \eqref{eq:DualObjective} increases by $D_j \alpha_j = D_j \vd_j = v_j$. The $\alpha_j$ variables are typically used to cover all the constraints of a completed job $j$, since the cost of covering their constraints is equal to their value.
The last variables $\beta_i(t)$ appear in all constraints associated with a server $i$ and time $t$. These variables are typically used to cover the dual constraints associated with incomplete jobs, since these variables are shared across the constraints of all jobs.

%
%
We denote by $OPT^*(\type)$ the optimal fractional solution of the dual program for an instance $\type$.
Define $\IntegralityGap(s) = \max_{\type : s(\type)=s} \left\{\ffrac{v(OPT^*(\type))}{v(OPT(\type))}\right\}$ as the integrality gap for instances with slackness $s$.
We are interested in online scheduling algorithms that induce upper bounds on the integrality gap.

\begin{definition}
An online scheduling algorithm $\Alg$ induces an upper bound on the integrality gap for a given slackness $s$ if $\IntegralityGap(s) \le \CompetitiveRatio_{\Alg}(s)$.
\end{definition}

The \emph{dual fitting} technique bounds both the competitive ratio $\CompetitiveRatio_\Alg(s)$ of an online algorithm $\Alg$ and the integrality gap $\IntegralityGap(s)$ by constructing a feasible solution to the dual program and bounding its dual cost.
Every feasible dual solution induces an upper bound on the optimal fractional solution, and the well-known weak duality theorem implies that $v(OPT(\type)) \le v(OPT^*(\type))$.
Moreover, $v(\Alg(\type)) \le v(OPT(\type))$. Therefore, we can obtain bounds on the integrality gap and the competitive ratio of $\Alg$. This is summarized in the following theorem.



\begin{theorem}[Dual Fitting \cite{Vazirani01}]
\label{thm:DualFitting}
Let $\Alg$ be an online scheduling algorithm.
If for every instance $\type$ with slackness $s=s(\type)$ there exists a feasible dual solution $(\alpha,\beta,\pi)$ with a dual cost of at most $r(s) \cdot v(\Alg(\type))$, then $\emph{\CompetitiveRatio}_{\Alg}(s) \,\le\, r(s)$ and $\emph{\IntegralityGap}(s) \,\le\, r(s)$.
\end{theorem}



\section{Truthful Non-Committed Scheduling}
\label{sec:NonCommitted_Truthfulness}

Our first goal is to design a truthful online scheduling mechanism under the deadline slackness assumption, without regard for commitments. The algorithmic version of this problem was studied in \cite{LMNY13}. \cite{LMNY13} presents a modified greedy scheduling algorithm, and shows that it obtains a
constant competitive ratio for any $s>1$.
%
%
However, the algorithm in \cite{LMNY13} is not monotone.
We refer the reader to the full version of the paper for a counterexample, in which a job that would not be completed can manipulate the algorithm by reporting a lower value and consequently be completed by its deadline.

In this section, we develop a new \emph{truthful} mechanism $\Alg_T$, which also obtains a constant competitive ratio for any $s > 1$.
The mechanism will be parameterized by constants $\thres > 1$ and $\gap > 1$, which will be specified below.
A key element in $\Alg_T$ is dividing the jobs into buckets (classes), differentiated by their value densities. Precisely, the job classes are $\class_\ell = \left\{ j \mid \vd_j \in \big[\thres^\ell , \thres^{\ell+1}\big) \right\}$.
Notice that job $j$ belongs to class $\class_\ell$ for $\ell = \lfloor \log_\thres (\vd_j)\rfloor$.  We think of a job $j'$ as dominating another job $j$ if $j'$ is in a ``higher'' bucket than $j$.  More formally, we use the following notation throughout the section:

\begin{definition}
Given jobs $j$ and $j'$, we say that $j' \succ j$ if $\lfloor \log_\thres (\vd_{j'})\rfloor > \lfloor \log_\thres (\vd_j)\rfloor$.
\end{definition}

At a high level, algorithm $\Alg_T$ proceeds as follows.
At each point in time, $\Alg_T$ will process the job with highest priority according to the ordering $\succ$.
That is, a pending job $j'$ can preempt a running job $j$ only if $j' \succ j$.  However, there is an important exception: if a job $j$ has not begun its execution by time $d_j - \gap D_j$, then the scheduler will discard that job and will not schedule it thereafter (i.e., it can be rejected immediately).
The following intuition motivates these principles. The preemption rule guarantees that the running jobs belong to the highest classes out of all available jobs (proven later, see Claim \ref{thm:NonCommitted_Truthful_SingleServer_Claim2}). This prevents users from benefiting from a misreport of their values.
The decision to not execute a job that has not begun by time $d_j - \gap D_j$ is used to bound the competitive ratio; note that this condition implies that there is slackness in the time interval from the first time the job is executed, to the job's deadline.


We now formally describe our truthful algorithm for the single server case (see Algorithm \ref{alg:NonCommitted_Truthful_SingleServer} for pseudo-code). The extension to multiple servers can be found in the full version of the paper.


\begin{algorithm}
\label{alg:NonCommitted_Truthful_SingleServer}
\SetKw{KwEvent}{Event:}
\DontPrintSemicolon
\caption{Truthful Non-Committed Algorithm $\Alg_T$ for a Single Server}

$\forall t, \,\,\,\,\, J^P(t) = \left\{ \,j\in\JobInput \mid j \textrm{ partially processed by } \Alg_T \textrm{ at time } t \,\wedge\, t\in [a_j,d_j]  \right\}$.\;
$\,\,\,\,\,\,\,\,\,\,\,\,\, J^E(t) = \left\{ \,j\in\JobInput \mid j \textrm{ unallocated by } \Alg_T \textrm{ at time } t \,\wedge\, t\in [a_j,d_j - \gap D_j]  \right\}$.\;
\;

\KwEvent On arrival of job $j$ at time $t=a_j$:\;
$\,\,\,$ 1. call ClassPreemptionRule($t$). \;
\;

\KwEvent On completion of job $j$ at time $t$:\;
$\,\,\,$ 1. resume execution of job $j' = \arg\max \left\{ \vd_{j'} \mid j' \in J^P(t) \right\}$.\;
$\,\,\,$ 2. call ClassPreemptionRule($t$).\;
$\,\,\,$ 3. delay the output response of $j$ until time $d_j$.\;
\;

\textbf{ClassPreemptionRule ($t$):}\;
$\,\,\,$ 1. $j \,\,\, \leftarrow$ job currently being processed.\;
$\,\,\,$ 2. $j^* \leftarrow \arg\max \left\{ \vd_{j^*} \mid j^{*} \in J^E(t) \right\}$.\;
$\,\,\,$ 3. if $\left( j^* \succ j \right):$\;
$\,\,\,\,\,\,\,\,\,\,\,\,\,\,\,\,\,\,$ 3.1. preempt $j$ and run $j^*$.\;

\end{algorithm}


Note that the algorithm maintains two job sets.
The first set $J^P(t)$ represents jobs $j$ that have been partially processed by time $t$ and can still be executed.  The second set $J^E(t)$ represents all jobs $j$ that have not been allocated by time $t$, where $t \le d_j - \gap D_j$.

The algorithm's decisions are triggered by one of the following two events: either when a new job arrives, or when a processed job is completed.
The algorithm handles both events similarly.
When a new job $j$ arrives, the algorithm invokes a \emph{class preemption rule}, which decides which job to process. In this case, the arriving job $j$ preempts the running job only if it belongs to a higher class.
The second type of event occurs when the running job is completed. As mentioned earlier, the algorithm delays the output of the job until its respective deadline (line 3). 
When a job is completed, the algorithm resumes the best job $j'$ among the preempted jobs in $J^P(t)$ (line 1) and calls the class preemption rule (line 2). The class preemption rule would override the decision to resume $j'$ if there exists an unallocated job $j^*$ in $J^E(t)$ belonging to a higher class. In that case, $j^*$ is processed and $j'$ remains preempted.
Notice that in both cases, the algorithm favors jobs belonging to higher classes. Formally,
\begin{claim}
\label{thm:NonCommitted_Truthful_SingleServer_Claim2}
Let $j = j_{\Alg_T}(t)$ be the job processed at time $t$ by $\Alg_T$. Let $j' \in J^P(t) \,\cup\, J^E(t)$. That is, $j$ has either been allocated by time $t$ and $t\in[a_{j'},d_{j'}]$, or $j$ has not been allocated by time $t$ and $t\in [a_{j'},d_{j'} - \gap D_j]$. Then, $j' \not\succ j$.
\end{claim}
\begin{proof}
Assume towards contradiction that $j' \succ j$. Let $t^*$ denote the earliest time job inside the interval $\big[a_{j'},t\big]$ during which $j$ is allocated. Note that $t^*$ must exist, since the claim assumes that $j$ has being processed at time $t$. At time $t^*$, the algorithm $\Alg$ either started processing $j$ or resumed the execution of $j$. For $\Alg$ to start $j$, the threshold preemption rule must have preferred $j$ over $j'$, which is impossible. The second case where $\Alg$ resumed the execution of job $j$ is also impossible, since either $j'$ would have been resumed instead of $j$, or the threshold preemption rule would have immediately preempted $j$.
We conclude that $j' \not\succ j$.
\end{proof}
Claim \ref{thm:NonCommitted_Truthful_SingleServer_Claim2} implies that at any point in time, the job allocated by $\Alg_T$ belongs to the highest class among the jobs that can be processed, i.e., either an unallocated job $j$ such that $t \in [a_j, d_j - \gap D_j ]$ or a partially processed job $j$ such that $t \in [a_j,d_j]$. Notice further that equalities in job classes are broken in favor of partially processed jobs. This feature is crucial for proving the truthfulness and the performance guarantees of our algorithm. Using Claim \ref{thm:NonCommitted_Truthful_SingleServer_Claim2} we prove an additional property, which is also required for establishing truthfulness.

\begin{claim}
\label{thm:NonCommitted_Truthful_SingleServer_Claim4}
At any time $t$, the set $J^P(t)$ contains at most one job from each class.
\end{claim}

\begin{proof}
By induction. Assume the claim holds and consider one of the possible events. Upon arrival of a new job $j^*$ at time $t$, the threshold preemption rule allocates $j^*$ only if $j^* \succ j$. Since $j$ is the maximal job in $J^P(t)$, with respect to $\succ$, if $j^*$ is allocated then it is the single job in $J^P(t)$ from its class. Upon completion of job $j$, it is removed from $J^P(t)$ and the threshold preemption rule is invoked. As before, if a new job is allocated, it belongs to a unique class.
\end{proof}

We now prove that $\Alg_T$ is truthful, i.e., $\Alg_T$ can be used to design a truthful online scheduling mechanism.

\begin{claim}
\label{thm:NonCommitted_Truthful_SingleServer_Truthfulness}
The algorithm $\Alg_T$ (single server) is monotone.
\end{claim}

The full proof of Claim \ref{thm:NonCommitted_Truthful_SingleServer_Truthfulness} appears in Appendix \ref{sec:Appendix_NonCommitted_Truthful_SingleServer}. The intuition behind the result is as follows.
The algorithm is defined so that the processing of higher-class jobs is independent of the presence of lower-class jobs in the system. As a result, a job $j$ is completed if precisely two conditions hold: first, that there is some time in $[a_j, d_j - \gap D_j]$ in which no job of equal or higher class is executing (so that job $j$ can start), and second, there are at least $D_j$ units of time after the earliest such start time, but before $d_j$, in which higher class jobs are not executing. These conditions are well-defined because the processing of job $j$ does not impact the times in which jobs of higher class are processed. One can then note, however, that each of these two conditions are monotone with respect to the job's class, length, arrival time, and deadline.  One can therefore conclude that the algorithm is monotone, and hence truthfulness follows from Theorem \ref{thm:truth}.

The competitive-ratio analysis of $\Alg_T$ is similar to the analysis of the non-truthful algorithm $\Alg$ \cite{LMNY13}, and proceeds via the dual fitting methodology. 
The full proof is described in Appendix \ref{sec:Appendix_NonCommitted_Truthful_SingleServer}.
Our result is the following.

\begin{theorem}
\label{thm:NonCommitted_Truthful_SingleServer}
The mechanism $\Alg_T$ (single-server) is truthful and obtains a competitive ratio
\begin{eqnarray*}
\emph{\CompetitiveRatio}_{\Alg_T}(s) & = & 2 + \Theta\left(\frac{1}{\sqrt[3]{s}-1}\right) + \Theta\left(\frac{1}{(\sqrt[3]{s}-1)^3}\right), \quad s>1.
\end{eqnarray*}
\end{theorem}

\subsection{Extension to Multiple Servers}

We next extend our algorithm to handle multiple servers. We provide a high level description of the algorithm; 
the details can be found in Appendix \ref{sec:Appendix_NonCommitted_Truthful_MultipleServers}.
The multiple server algorithm runs a local copy of the single server algorithm on each of the $C$ servers. The algorithm allows a job to use different servers throughout time (equivalently, we use say that a job is allowed to \emph{migrate} between servers), yet with some restrictions: a preempted job $j$ can migrate to any other server before time $d_j - \gap D_j$. After that time, the job may only use the subset of servers which were allocated to it before time $d_j - \gap D_j$.
We obtain the following competitive-ratio result.
\begin{theorem}
\label{thm:NonCommitted_Truthful_MultipleServers}
The algorithm $\Alg_T$ (multiple-servers) obtains a competitive ratio of:
\begin{eqnarray*}
\emph{\CompetitiveRatio}_{\Alg_T}(s) & = & 2 + \Theta\left(\frac{1}{\sqrt[3]{s}-1}\right) + \Theta\left(\frac{1}{(\sqrt[3]{s}-1)^3}\right), \quad s>1.
\end{eqnarray*}
\end{theorem}
Observe that the competitive ratio for the multiple server case is (asymptotically) identical to the bound obtained for a single server. However, we note that the constants hidden inside $\Theta$ are slightly larger for the multiple-server case.

\section{Committed Scheduling}
\label{sec:Committed}

In this section we develop the first committed (i.e., responsive) scheduler for online scheduling with general job types, assuming deadline slackness. Our solution is based on a novel reduction of the problem to the ``familiar territory" of non-committed scheduling. 
We introduce a parameter $\omega \in (0,1)$ that affects the time by which the scheduler commits.
Specifically, the scheduler we propose decides whether to admit jobs during the first $(1-\omega)$-fraction of their availability window, i.e., by time $d_j - \omega (d_j - a_j)$ for each job $j$.
The deadline slackness assumption ($d_j - a_j \ge sD_j$) then implies that our scheduler is $(\omega s)$-responsive (cf.~Definition \ref{def:betaResponsive} for $\beta = \omega s$).

We start with the single server case (Section \ref{sec:Committed_SingleServer_Algorithm}), where we highlight the main mechanism design principles. We then extend our solution to accommodate multiple servers, which requires some subtle changes in our proof methodology (Section \ref{sec:Committed_MultipleServers}).

Our competitive-ratio results hold for slackness values greater than some threshold (e.g., $s>4$ for the single-server case). In Section \ref{sec:Committed_LowerBound}, we provide an indication that high slackness is indeed required, by obtaining a related impossibility result for inputs with small slackness.


\subsection{Reduction for a Single Server}
\label{sec:Committed_SingleServer_Algorithm}

Our reduction consists of two key components: (1) \emph{simulator}: a virtual server used to simulate an execution of a non-committed algorithm $\Alg$; and (2) \emph{server}: the real server used to process jobs. The speeds of the simulator and server are the same. We emphasize that the simulator does not utilize actual job resources. It is only used to determine which jobs to admit.
We use the simulator to simulate an execution of the non-committed algorithm. Upon arrival of a new job, we submit the job to the simulator with a \emph{virtual type}, defined below. If a job is completed on the simulator, then the committed scheduler admits it to the system and processes it on the server (physical machine). We argue later that the overall value gained by the algorithm is relatively high, compared to the value guaranteed by $\Alg$.

We pause briefly to highlight the challenges in such simulation-based approach. The underlying idea is to admit and process jobs on the server only after they are ``virtually" completed by $\Alg$ on the simulator. If the simulator completes all jobs near their actual deadlines, the scheduler might not be able to meet its commitments. This motivates us to restrict the latest time in which a job can be admitted. The challenge is to guarantee that all admitted jobs are completed, while still guaranteeing relatively high value.

We now provide more details on how the simulator and server are handled by the committed scheduler throughout execution.
\\

\noindent
\textbf{Simulator.}
The simulator runs an online non-committed scheduling algorithm $\Alg$.
Every arriving job $j$ is automatically sent to the simulator with a virtual type $\virtual{\type}_j = \langle v_j, \virtual{D}_j, a_j, \virtual{d}_j\rangle$, where $\virtual{d}_j = d_j - \omega (d_j - a_j)$ is the virtual deadline of $j$, and $\virtual{D}_j = \ffrac{D_j}{\omega}$ is the virtual demand of $j$.
If $\Alg$ completes the virtual request of job $j$ by its virtual deadline, then $j$ is admitted and sent to the server.
\\

\noindent
\textbf{Server.}
The server receives admitted jobs once they have been completed by the simulator, and processes them according to the Earliest Deadline First (EDF) allocation rule. That is, at any time $t$ the server processes the job with the earliest deadline out of all admitted jobs that have not been completed.
\\

The reduction effectively splits the availability window to two subintervals.  The first $(1-\omega)$ fraction is the first subinterval and the remainder is the second. 
The virtual deadline $\virtual{d}_j$ serves as the breakpoint between the two intervals.
During the first subinterval, the algorithm uses the simulator to decide whether to admit $j$ or not.  Then, at time $\virtual{d}_j$, it communicates the decision to the job. 
In practical settings, this may allow a rejected job to seek other processing alternatives during the remainder of the time. Furthermore, if $j$ is admitted, the scheduler is left with at least $\omega(d_j - a_j)$ time to process the admitted job on the server.

The virtual demand of each job $j$ is increased to $\ffrac{D_j}{\omega}$. We use this in our analysis to guarantee that the server meets the deadlines of admitted jobs. Note that we must require $\ffrac{D_j}{\omega} \le (1-\omega)sD_j$, otherwise $j$ could not be completed on the simulator. By rearranging terms, we get a constraint on the values of $s$ for which our algorithm is feasible: $s \ge \frac{1}{\omega(1-\omega)}$.



\subsubsection{Correctness}
\label{sec:Committed_SingleServer_Correctness}
We now prove that when the reduction is applied, each accepted job is guaranteed to finish by its deadline.
Note that the simulator can complete a job before its virtual deadline, hence it may be  admitted earlier. However, in the analysis below, we assume without loss of generality that jobs are admitted at their virtual deadline. Accordingly,
We define the \emph{admitted type} of job $j$ as $\admitted{\type}_j = \langle v_j, D_j, \virtual{d}_j, d_j\rangle$.

Recall that $\Alg_C(\type)$ represents the jobs completed by the committed algorithm. Equivalently, these are the jobs completed by the non-committed algorithm $\Alg$ on the simulator.
To prove that $\Alg_C$ can meet its guarantees, we must show that the EDF rule deployed by the server completes all jobs in $\Alg_C(\type)$, when submitted with their admitted types.
It is well known that for every set of jobs $S$, if $S$ can be feasibly allocated on a single server (i.e., before their deadline), then EDF produces a feasible schedule of $S$. Hence, it suffices to prove that there exists a feasible schedule of $\Alg_C(\type)$.
We prove the following general claim, which implies the correctness of our algorithm.

\begin{theorem}
\label{thm:Committed_SingleServer_Correctness}
Let $S$ be a set of jobs. For each job $j\in S$, define the virtual deadline of $j$ as $\virtual{d}_j = d_j - \omega(d_j - a_j)$.
If there exists a feasible schedule of $S$ on a single server with respect to the virtual types $\virtual{\type}_j = \big\langle v_j, \ffrac{D_j}{\omega},a_j,\virtual{d}_j\big\rangle$ for each $j\in S$,
then there exists a feasible schedule of $S$ on a single server with respect to the admitted types $\admitted{\type}_j = \big\langle v_j,D_j,\virtual{d}_j,d_j\big\rangle$ for each $j\in S$.
\end{theorem}

\begin{proof}
We describe an allocation algorithm that generates a feasible schedule of $S$ with respect to admitted types. That is, the algorithm produces a schedule where a each job $j\in S$ is processed for $D_j$ time units inside the time interval $[\virtual{d}_j, d_j]$.
The algorithm we describe allocates jobs in decreasing order of their virtual deadlines. For two jobs $j,j'\in S$, we write $j' \succ j$ when $\virtual{d}_{j'} > \virtual{d}_j$. In each iteration, the algorithm considers some job $j\in S$ by the order induced by $\succ$, breaking ties arbitrarily. We say that time $t$ is \emph{used} when considering $j$ if the algorithm has allocated some job $j'$ at time $t$; otherwise, we say that $t$ is \emph{free}. We denote by $\mathcal{U}_j$ and $\mathcal{F}_j$ the set of used and free times when the algorithm considers $j$, respectively.
The algorithm works as follows. Consider an initially empty schedule.
We iterate over jobs in $S$ in decreasing order of their virtual deadlines, breaking ties arbitrarily; this order is induced by $\succ$. Each job $j$ in this order is allocated during the latest possible $D_j$ free time units.
Formally, define $t' = \arg\max\{ t : \left| [t,d_j] \cap \mathcal{F}_j \right| = D_j \}$ as the latest time such that there are exactly $D_j$ free time units during $[t', d_j]$. The algorithm allocates $j$ during those free $D_j$ time units $[t',d_j] \cap \mathcal{F}_j$.

We now prove that the algorithm returns a feasible schedule of $S$, with respect to the admitted job types.
It is enough to show that when a job $j\in S$ is considered by the algorithm, there is enough free time to process it; namely, there should be at least $D_j$ free time units during $[\virtual{d}_j,d_j]$. Consider the point where the algorithm allocates a job $j\in S$.
Define $\ell_R = \max\{ \ell \mid [d_j,d_j + \ell] \subseteq \mathcal{U}_j \}$ and denote $t_R = d_j + \ell_R$.
By definition, the time interval $[d_j,t_R]$ is the longest continuous block that starts at $d_j$ in which all times $t\in [d_j,t_R]$ are used.
Define $t_L = a_j - \ell_R \cdot \ffrac{(1-\omega)}{\omega}$.
We claim that any job $j' \succ j$ allocated in the interval
$[\virtual{d}_j,t_R]$ must satisfy $[a_{j'},d_{j'}] \subseteq [t_L,t_R]$. Assume the claim holds. We show how the claim leads to the theorem. Denote by $J_{LR}$ all jobs $j' \succ j$ that have been allocated sometime during the interval $[\virtual{d}_j,t_R]$. Obviously, we also have $[a_j,d_j] \subseteq [t_L,t_R]$.
Now, since we know there exists a feasible schedule of $S$ with respect to the virtual types, we can conclude that the total virtual demand of jobs in $J_{LR} \cup \{j\}$ is at most $t_R - t_L$, since the interval $[t_L,t_R]$ contains the availability windows of all these jobs.
Notice that $t_R - t_L = \ffrac{(t_R - \virtual{d}_j)}{\omega}$.
Since the virtual demand is $\ffrac{1}{\omega}$ times larger than the admitted demand, we can conclude that the total amount of used time slots during $[\virtual{d}_j,t_R]$ is at most $(t_R - \virtual{d}_j) - D_j$.
Thus, there have to be $D_j$ free time units during $[\virtual{d}_j,d_j]$ since $[d_j,t_R]$ is completely full.
%
It remains to prove the claim. Let $j' \in J_{LR}$.
Notice that $d_{j'} \le t_R$; otherwise, the allocation algorithm could have allocated $j'$ after time $t_R$, and since we assume $j'$ has been allocated sometime between $[\virtual{d}_j,d_j]$, this would contradict the definition of $t_R$. Also, $j' \succ j$ means $\virtual{d}_{j'} \ge \virtual{d}_j$. Therefore:
\begin{eqnarray*}
a_{j'}
   & \,\,=\,\, &
\frac{1}{\omega} \cdot \virtual{d}_{j'} - \frac{1-\omega}{\omega} \cdot d_{j'}
   \,\,\,\, \ge \,\,\,\,
\frac{1}{\omega} \cdot \virtual{d}_j - \frac{1-\omega}{\omega} \cdot t_R
\\
   & \,\,=\,\, &
\frac{1}{\omega} \cdot d_j - (d_j - a_j) - \frac{1-\omega}{\omega} \cdot (d_j + \ell_R)
   \,\,\,\,\, = \,\,\,\,\,
a_j - \frac{1-\omega}{\omega} \cdot \ell_R
   \,\,\,\, = \,\,\,\,
t_L
\end{eqnarray*}
which completes the proof.
\end{proof}

\begin{figure}
  \begin{center}
  \includegraphics[width=\textwidth]{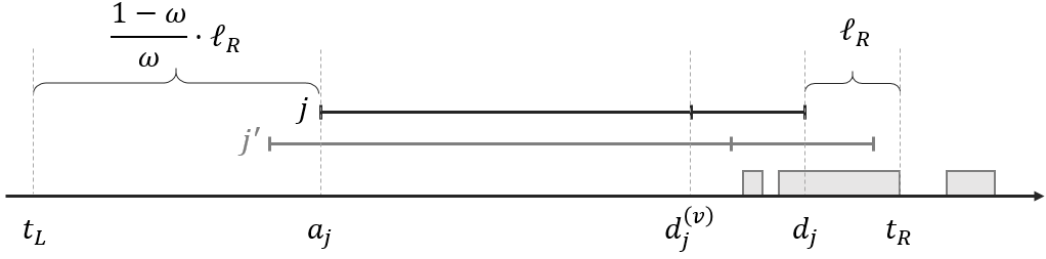}
  \caption{Illustration of the proof to Theorem \ref{thm:Committed_SingleServer_Correctness}.}
  \label{fig:CommittedProof}
  \end{center}
\end{figure}


\subsubsection{Competitive Ratio}
\label{sec:Committed_SingleServer_CompetitiveRatio}

We now analyze the competitive ratio obtained via the single server reduction. The competitive ratio is bounded using dual fitting arguments. Specifically, for every instance $\type$ with slackness $s=s(\type)$, we construct a feasible dual solution $(\alpha,\beta)$ with dual cost proportional to $v(\Alg_C(\type))$, the total value gained by $\Alg_C$ on $\type$.
Recall the dual constraints \eqref{eq:DualConstraint} corresponding to types $\tau_j = \big\langle v_j, D_j, a_j, d_j\big\rangle$. For the single server case, we make two simplifications. First, we denote $\beta(t) = \beta_1(t)$ to simplify notation. Second, we assume that $\pi = 0$ without loss of generality\footnote{This assumption is valid due to the redundancy of the primal constraints corresponding to $\pi$ for a single server.}. The dual constraints corresponding to $\type$ reduce to:
\begin{alignat}{5}
& \quad \,\,\,\alpha_j + \beta(t) \,\,\ge\,\, \vd_j & \,\,\,\,\, & \quad \forall j\in\JobInput ,\, t\in \big[a_j,d_j\big]. \label{eq:DualConstraintSingleServer}
\end{alignat}
Our goal is to construct a dual solution which satisfies \eqref{eq:DualConstraintSingleServer} and has a dual cost of at most $r \cdot v(\Alg_C(\type))$ for some $r$. Note that $v(\Alg_C(\type)) = v(\Alg(\virtual{\type}))$. To do so, we transform a dual solution corresponding to virtual types $\virtual{\type}$ to a dual solution satisfying \eqref{eq:DualConstraintSingleServer}. The dual constraints corresponding to the virtual types are:
\begin{alignat}{5}
& \quad \,\,\,\alpha_j + \beta(t) \,\,\ge\,\, \omega \vd_j & \,\,\,\,\, & \quad \forall j\in\JobInput ,\, t\in \big[a_j,\virtual{d}_j\big] \label{eq:DualConstraintVirtualSingleServer}
\end{alignat}
Assume that the non-committed algorithm $\Alg$ induces an upper bound on $\IntegralityGap(\virtual{s})$, where $\virtual{s} = s \cdot \omega(1-\omega)$ is the slackness of the virtual types $\virtual{\type}$. This implies that the optimal dual solution $(\alpha^*,\beta^*)$ satisfying \eqref{eq:DualConstraintVirtualSingleServer} has a dual cost of at most
$\CompetitiveRatio_{\Alg}(\virtual{s}) \cdot v(\Alg(\virtual{\type})) = \CompetitiveRatio_{\Alg}(\virtual{s}) \cdot v(\Alg_{C}(\type))$.
Yet, $(\alpha^*,\beta^*)$ satisfies \eqref{eq:DualConstraintVirtualSingleServer}, while we require a solution that satisfies \eqref{eq:DualConstraintSingleServer}.
To construct a feasible dual solution corresponding to the original job types $\type$, we perform two transformations on $(\alpha^*,\beta^*)$ called \emph{stretching} and \emph{resizing}.

%
%
\begin{lemma}[Resizing Lemma]
\label{thm:ResizingLemma}
Let $(\alpha,\beta)$ be a feasible solution for the dual program corresponding to a type profile $\type_j = \big\langle v_j,D_j,a_j,d_j\big\rangle$. There exists a feasible solution $(\alpha',\beta')$ for the dual program with demands $D'_j = f\cdot D_j$ for some $f > 0$, with a dual cost of:
\begin{eqnarray*}
\sum_{j\in\JobInput} D'_j \alpha'_j  \,+\, \intop_{0}^{\infty} \beta'(t) dt
   & = &
\sum_{j\in\JobInput} D_j \alpha_j  \,+\,  \frac{1}{f} \cdot \intop_{0}^{\infty} \beta(t) dt.
\end{eqnarray*}
\end{lemma}

\begin{proof}
Notice that the value density corresponding to $D'_j = f\cdot D_j$ is $\vd'_j = \ffrac{\vd_j}{f}$. Hence, by setting $\alpha'_j = \ffrac{\alpha_j}{f}$ for every job $j\in\JobInput$ and $\beta(t) = \ffrac{\beta(t)}{f}$ for every time $t$, we obtain a feasible dual solution corresponding to resized demands $D'_j$. The dual cost is as stated since $D'_j \alpha'_j = D_j \alpha_j$ for every job $j$.
\end{proof}

%
%
\begin{lemma}[Stretching Lemma, \cite{LMNY13}]
\label{thm:StretchingLemma}
Let $(\alpha,\beta)$ be a feasible solution for the dual program corresponding to a type profile $\type_j = \big\langle v_j,D_j,a_j,d_j \big\rangle$. There exists a feasible solution $(\alpha',\beta')$ for the dual program with deadlines $d'_j = d_j + f\cdot (d_j - a_j)$ for some $f$, with a dual cost of:
\begin{eqnarray*}
\sum_{j\in\JobInput} D_j \alpha'_j  \,+\, \intop_{0}^{\infty} \beta'(t) dt
   & = &
\sum_{j\in\JobInput} D_j \alpha_j  \,+\,  (1+f) \cdot \intop_{0}^{\infty} \beta(t) dt.
\end{eqnarray*}
\end{lemma}

These two lemmas allow us to bound the competitive ratio of $\Alg_C$.


\begin{theorem}
\label{thm:Committed_SingleServer_CompetitiveRatio}
Let $\Alg$ be a single server scheduling algorithm that induces an upper bound on the integrality gap $\emph{\IntegralityGap}(\virtual{s})$ for $\virtual{s} = s \cdot \omega(1 - \omega)$ and $\omega\in(0,1)$. Let $\Alg_C$ be the committed algorithm obtained by the single server reduction. Then $\Alg_C$ is $\omega s$-responsive and
\begin{eqnarray*}
\emph{\CompetitiveRatio}_{\Alg_C}(s) & \le & \frac{\emph{\CompetitiveRatio}_{\Alg}\Big( s \cdot \omega(1 - \omega) \Big)}{\omega(1-\omega)}
\,\,\,\,\,\,\, , \,\,\,\,\,\,\, s > \frac{1}{\omega (1-\omega)}.
\end{eqnarray*}
\end{theorem}

\begin{proof}
We first prove that the scheduler is $\omega s$-responsive. Note that each job $j$ is either committed or rejected by its virtual deadline $\virtual{d}_j = d_j - \omega(d_j - a_j)$. The deadline slackness assumption states that $d_j - a_j \ge sD_j$ for every job $j$. Hence, each job is notified by time $d_j - \omega s D_j$, as required.

We now bound the competitive ratio. Consider an input instance $\type$ and denote its slackness by $s=s(\type)$.
Let $\virtual{\type}$ denote the virtual types corresponding to $\type$, and let $\virtual{s}= s \cdot \omega(1-\omega)$ denote their slackness.
We prove the theorem by constructing a feasible dual solution $(\alpha,\beta)$ satisfying \eqref{eq:DualConstraintSingleServer} and bounding its total cost.
By the assumption on $\Alg$, the optimal fractional solution $(\alpha^*,\beta^*)$ corresponding to $\virtual{\type}$ has a dual cost of at most $\CompetitiveRatio_{\Alg}(\virtual{s}) \cdot v(\Alg(\virtual{\type})) = \CompetitiveRatio_{\Alg}(\virtual{s}) \cdot v(\Alg_C(\type))$. We transform $(\alpha^*,\beta^*)$ into a feasible solution $(\alpha,\beta)$ corresponding to $\type$ by applying the resizing lemma and the stretching lemma, as follows.
\begin{itemize}
  \item We first apply the resizing lemma for $f=\frac{1}{\omega}$ to cover the increased job demands during simulation. The dual cost increases by a multiplicative factor of $\frac{1}{\omega}$.
  \item We then apply the stretching lemma to cover the remaining constraints; that is, the times in the jobs' execution windows not covered by the execution windows of the virtual types. We choose $f$ such that $d_j = \virtual{d}_j + f \cdot \big( \virtual{d}_j - a_j\big)$; hence, $f = \frac{\omega}{1-\omega}$. As a result, the competitive ratio is multiplied by an additional factor of $1+f = \frac{1}{1 - \omega}$.
\end{itemize}
After applying both lemmas, we obtain a feasible dual solution that satisfies the dual constraints \eqref{eq:DualConstraintSingleServer}. The dual cost of the solution is at most
$\frac{1}{\omega(1-\omega)} \cdot \CompetitiveRatio_{\Alg}\big(s \cdot \omega (1 - \omega) \big) \cdot v(\Alg_{C}(\type))$. The theorem follows through the correctness of the dual fitting technique, Theorem \ref{thm:DualFitting}.
\end{proof}

Applying Theorem \ref{thm:Committed_SingleServer_CompetitiveRatio} to the single server scheduling algorithm $\Alg_T$ from Section \ref{sec:NonCommitted_Truthfulness} and choosing $\omega = 1/2$, one obtains a $(s/2)$-responsive scheduler with a competitive ratio that approaches $8$ as $s$ grows large.  However, we note that a more careful analysis, specific to the algorithm $\Alg_T$, leads to an improved bound (approaching $5$ as $s$ grows large).  This tighter analysis, which involves merging the dual-fitting techniques from Theorem \ref{thm:Committed_SingleServer_CompetitiveRatio} with the dual-fitting techniques used to bound the competitive ratio of $\Alg_T$, is described in Appendix \ref{sec:Appendix_Theorem_Statements}.




\subsection{Reductions for Multiple Servers}
\label{sec:Committed_MultipleServers}

We extend our single server reduction to incorporate multiple servers.
We distinguish between two cases based on the following definition.

\begin{definition}
A scheduler is called \emph{non-migratory} if it does not allow preempted jobs to resume their execution on different servers. That is, a job is allocated at most one server throughout its execution.
\end{definition}

Constant-competitive non-migratory schedulers are known to exist in the presence of deadline slackness \cite{LMNY13}. Given such a scheduler, we can easily construct a committed algorithm for multiple servers by extending the single server reduction; see full paper for details. 
However, we do not know how to use this reduction to obtain a committed scheduler which is truthful, since it requires that the non-committed scheduler is both truthful and non-migratory; unfortunately, we are not aware of such schedulers.

Therefore, we construct below a second reduction, which does not require a non-migratory non-committed scheduler. This is essential for Section \ref{sec:Committed_Truthful}, where we design a truthful committed scheduler.
We note that the first reduction leads to better competitive-ratio guarantees, hence should be preferred in domains where users are not strategic.

\subsubsection{Non-Migratory Case}
\label{sec:Committed_MultipleServers_NoMigration}
In the following, let $\Alg$ be a non-committed scheduler for multiple servers which is non-migratory. We extend our single server reduction to obtain a committed scheduler $\Alg_C$ for multiple servers. The reduction remains essentially the same: the simulator runs the non-committed scheduler on a system with $C$ virtual servers. When a job is completed on virtual server $i$, it is admitted and processed on server $i$. Each server runs the EDF rule on the jobs admitted to it. To prove correctness (i.e., the scheduler meets all commitments), we simply apply Theorem \ref{thm:Committed_SingleServer_Correctness} on each server independently. The bound on the competitive ratio obtained in Theorem \ref{thm:Committed_SingleServer_CompetitiveRatio} can be extended directly to the non-migratory model.

\begin{corollary}
\label{thm:Committed_MultipleServer_NonMigration_CompetitiveRatio}
Let $\Alg$ be a multiple server, non-migratory scheduling algorithm that induces an upper bound on the integrality gap $\emph{\IntegralityGap}(\virtual{s})$ for $\virtual{s} = s \cdot \omega(1 - \omega)$ and $\omega\in(0,1)$. Let $\Alg_C$ be the committed algorithm obtained by the multiple server reduction for non-migratory schedulers. Then $\Alg_C$ is $\omega s$-responsive and
\begin{eqnarray*}
\emph{\CompetitiveRatio}_{\Alg_C}(s) & \le & \frac{\emph{\CompetitiveRatio}_{\Alg}\Big( s \cdot \omega(1 - \omega) \Big)}{\omega(1-\omega)}
\,\,\,\,\,\,\, , \,\,\,\,\,\,\, s > \frac{1}{\omega (1-\omega)}.
\end{eqnarray*}
\end{corollary}

Applying Corollary \ref{thm:Committed_MultipleServer_NonMigration_CompetitiveRatio} to the non-migratory multiple-server algorithm presented in \cite{LMNY13} and setting $\omega = 1/2$, one obtains a $(s/2)$-responsive scheduling algorithm for multiple servers with competitive ratio $8 + \Theta\Big(\frac{1}{\sqrt[3]{s/4}-1}\Big) + \Theta\Big(\frac{1}{(\sqrt[3]{s/4}-1)^2}\Big)$.  As in Theorem \ref{thm:Committed_SingleServer_CompetitiveRatio}, one can achieve a tighter approximation factor of $5 + \Theta\Big(\frac{1}{\sqrt[3]{s/4}-1}\Big) + \Theta\Big(\frac{1}{(\sqrt[3]{s/4}-1)^2}\Big)$ using the details of the dual-fitting analysis from \cite{LMNY13}.  This gives the result described in Section~\ref{sec:Introduction_OurResults} as Theorem~\ref{thm.alg.responsive}.  The details of this improved analysis appear in Appendix \ref{sec:Appendix_Theorem_Statements}.

\subsubsection{Migratory Case}
\label{sec:Committed_MultipleServers_Migration}
We now assume that $\Alg$ allows migrations.  This will be important for truthful committed scheduling, explored in the next section.  Unfortunately, the reduction proposed for the non-migratory case does not work here. We explain why: consider some job $j$ that is admitted after being completed on the simulator; note that $j$ may have been processed on more than one virtual server. Our goal is to process $j$ by time $d_j$.
Assume each server runs the EDF rule on the jobs assigned to it, as suggested in Section \ref{sec:Committed_MultipleServers_NoMigration}.
Since $j$ has been processed on more than one virtual server, it is unclear how to assign $j$ to a server in a way that guarantees the completion of all admitted jobs.
One might suggest to assign each server $i$ the portion of $j$ that was processed on virtual server $i$. However, this does not necessarily generate a legal schedule. If each server runs EDF independently, a job might be allocated simultaneously on more than one server.

We propose the following modifications.
First, we will use a result of \cite{CLT05}, which shows that any set $\mathcal{S}$ of jobs that can be scheduled with migration on $C$ servers can also be scheduled without migration on $C$ servers with a speedup of $(3+2\sqrt{2}) \approx 5.828$.  Thus, if we increase the virtual demand of the jobs submitted to the simulator by this amount, then it will be possible to modify the resulting migratory schedule to be non-migratory.
Next, instead of running the EDF rule on each server independently, we run a global EDF rule. That is, at each time $t$ the system processes the (at most) $C$ admitted jobs with earliest deadlines. This is known as the EDF rule for multiple servers (also known as f-EDF \cite{Funk04}).
It is well known that the EDF rule is not optimal on multiple servers; formally, for a set $\mathcal{S}$ of jobs that can be feasibly scheduled on $C$ servers with migration, EDF does not necessarily produce a feasible schedule on input $\mathcal{S}$ \cite{HL89}. Nevertheless, it is known that EDF produces a feasible schedule of $\mathcal{S}$ when the servers are twice as fast \cite{PSTW97}.
Thus, since server speedup is directly linked with demand inflation, if we double the virtual demand of the jobs submitted to the simulator, we are guaranteed that EDF would produce a feasible schedule for the admitted jobs.
We will therefore modify the virtual demand of each job submitted to the simulator. The virtual demand of each job $j$ will be increased to $2(3+2\sqrt{2})\cdot \ffrac{D_j}{\omega}$. The additional factor of $3+2\sqrt{2} \approx 5.828$ is necessary for correctness, which is established in the following theorem.


\begin{theorem}
\label{thm:Committed_MultipleServer_Migration_CompetitiveRatio}
Let $\Alg$ be a multiple server scheduling algorithm that induces an upper bound on the integrality gap $\emph{\IntegralityGap}(\virtual{s})$ for $\virtual{s} = s \cdot \omega(1 - \omega)$ and $\omega\in(0,1)$. Let $\Alg_C$ be the committed algorithm $\Alg_C$ obtained by the multiple server reduction. Then $\Alg_C$ is $\omega s$-responsive and
\begin{eqnarray*}
\emph{\CompetitiveRatio}_{\Alg_C}(s)
  & \le &
\frac{11.656}{\omega (1-\omega)} \cdot \emph{\CompetitiveRatio}_{\Alg}\left(s \cdot \frac{\omega(1 - \omega)}{11.656} \right)
 \,\,\,\,\,\,\, , \,\,\,\,\,\,\,  s > \frac{11.656}{\omega (1-\omega)}.
\end{eqnarray*}
\end{theorem}

\begin{proof}
Let $\mathcal{S}$ denote the set of jobs admitted by the committed algorithm $\Alg_C$ on an instance $\type$.
To prove correctness, we must show that there exists a feasible schedule in which each job $j\in\mathcal{S}$ is allocated $2D_j$ demand during $[\virtual{d}_j, d_j]$. If so, then \cite{PSTW97} implies that EDF completes all admitted jobs by their deadline. This follows since:
\begin{enumerate}
  \item There exists a feasible schedule of $\mathcal{S}$ with types $\langle v_j, \frac{11.656}{\omega} \cdot D_j, a_j, \virtual{d}_j \rangle$ on $C$ servers with migration. This is the ``simulator'' schedule produced by the non-committed algorithm $\Alg$.
  \item \cite{CLT05} proved that any set $\mathcal{S}$ of jobs that can be scheduled with migration on $C$ servers can also be scheduled without migration on $C$ servers with $5.828$-speedup. As a result, there exists a feasible non-migratory schedule of $\mathcal{S}$ with types $\langle v_j, \frac{2}{\omega} \cdot D_j, a_j, \virtual{d}_j \rangle$ on $C$ servers.
  \item By applying Theorem \ref{thm:Committed_SingleServer_Correctness} on each server separately, we obtain a feasible non-migratory schedule of $\mathcal{S}$ with types $\langle v_j, 2D_j, \virtual{d}_j, d_j \rangle$ on $C$ servers, as desired.
  \item Therefore, EDF produces a feasible schedule of the admitted jobs $\mathcal{S}$ with types $\langle v_j, D_j, \virtual{d}_j, d_j \rangle$.
\end{enumerate}
We note that step 4 (i.e., using EDF) is necessary.  Even though Steps 2 and 3 establishe that feasible non-migratory schedules of $\mathcal{S}$ exist, they cannot necessarily be generated online, unlike EDF.  The competitive ratio can be bounded by following the same steps as in the single server case (Theorem \ref{thm:Committed_SingleServer_CompetitiveRatio}), however the resizing lemma must be applied with $f=11.656\omega$.
Finally, note that the slackness $s$ must satisfy $s(1-\omega) \ge \ffrac{11.656}{\omega}$, otherwise jobs could not be completed on the simulator.
\end{proof}

We use this reduction in Section \ref{sec:Committed_Truthful} to design a truthful committed scheduler for multiple servers.

\subsection{Impossibility Result}
\label{sec:Committed_LowerBound}

The committed schedulers we construct guarantee a constant competitive ratio, provided that the deadline slackness $s$ is sufficiently large. For example, $s$ has to be at least $(\omega(1-\omega))^{-1}$ for the single server case, implying that $s>4$ (since $\omega=1/2$ minimizes the expression). A valid question is whether these conditions on $s$ are merely a consequence of our choice of construction, or an inherent property of any possible committed scheduler.
In this subsection, we provide some indication that the latter is more likely, by provider an impossibility result. In particular, we prove a lower bound for committed schedulers that satisfy an additional requirement, termed \emph{no early processing}. A no early processing scheduler is a scheduler that may not process jobs before committing to their execution. We note that the schedulers we have designed in this section satisfy this requirement. It is also worth mentioning that although we did not include no-early processing as part of our $\beta$-responsive commitment definition, this is a natural property to require in many practical settings; e.g., when there is a cost (of data transmission, etc.) associated with beginning the execution of a job.  
Our result is the following.
\begin{theorem}
\label{thm:Committed_LowerBound}
Consider a cluster with $C<4$ machines. Then any committed scheduler that satisfies the no-early processing requirement has an unbounded competitive ratio for $s<\ffrac{4}{C}$.
\end{theorem}
In view of Theorem \ref{thm:Committed_SingleServer_CompetitiveRatio}, note that this bound is tight for the single server case (under the no early processing requirement).
%
It remains an open question whether removing the no-early processing requirement could lead to bounded competitive ratio for a larger range of $s$. More generally, obtaining tighter lower bounds for multiple servers is a direction that is still unresolved.



\section{Truthful Committed Scheduling}
\label{sec:Committed_Truthful}

In this section we construct a scheduling mechanism that is both truthful and committed. As it turns out, the reductions presented in the previous section preserve monotonicity with respect to values, deadlines, and demands, but not necessarily with respect to arrival times. Therefore, by plugging in an existing truthful non-committed scheduler (Section \ref{sec:NonCommitted_Truthfulness}), we can obtain a committed mechanism that is truthful assuming all arrival times are publicly known.
In Section \ref{sec:Committed_Truthful_Full} we show how to modify the construction to achieve full truthfulness.


\subsection{Public Arrival Times}
\label{sec:Committed_Truthful_PublicArrivalTimes}

In this subsection, we consider the case where job arrival times are common knowledge, i.e., users cannot misreport the arrival times of their jobs.
To construct the partially truthful mechanism, we apply one of the reductions from committed scheduling to non-committed scheduling (Section \ref{sec:Committed}) on a truthful non-committed mechanism, which we denote by $\Alg_T$. We denote by $\Alg_{\tilde{T}C}$ the resulting mechanism. In the following, we prove that $\Alg_{\tilde{T}C}$ is \emph{almost} truthful: it is monotone with respect to values, deadlines, and demands, but not with respect to arrival times.


\begin{claim}
Let $\Alg_T$ be a truthful scheduling algorithm, and let $\Alg_{\tilde{T}C}$ be a committed mechanism obtained by applying one of the reductions from committed scheduling to non-committed scheduling (assume all required preconditions apply). Then, $\Alg_{\tilde{T}C}$ is monotone with respect to values, demands and deadlines.
\end{claim}

\begin{proof}
Recall that upon an arrival of a new job $j$, the job is submitted to $\Alg_T$ with a virtual type of $\virtual{\type}_j = \langle v_j, \alpha D_j, a_j, \virtual{d}_j \rangle$ for some constant $\alpha \ge 1$ (the constant differs between the reductions for a single server and for multiple servers). Also recall that $\virtual{d}_j = d_j - \omega (d_j - a_j)$ is the virtual deadline of job $j$, which is a monotone function of $d_j$. Moreover, $\Alg_{\tilde{T}C}$ then completes job $j$ on input $\type$ precisely if $\Alg_T$ completes job $j$ on input $\virtual{\type}$. But since $\Alg_T$ is monotone, and since $v_j$, $\alpha D_j$, and $\virtual{d}_j$ are appropriately monotone functions of $v_j$, $D_j$, and $d_j$ (respectively), it follows that $\Alg_{\tilde{T}C}$ is monotone with respect to $v_j$, $D_j$, and $d_j$.
\end{proof}

\noindent
Hence, the reductions from committed to non-committed scheduling (Theorems \ref{thm:Committed_SingleServer_CompetitiveRatio} and \ref{thm:Committed_MultipleServer_Migration_CompetitiveRatio}) can be extended to guarantee truthfulness (public arrival times), as long as the given (non-committed) scheduler is monotone.

Recall that the definition of $\beta$-responsiveness for mechanisms requires not only that allocation decisions be made sufficiently early, but also that requisite payments be calculated in a timely fashion as well.  To obtain a $\beta$-responsive mechanism we must therefore establish that it is possible to compute payments at the time of commitment, for each job $j$.  Fortunately, because the time of commitment is independent of a job's reported value, this is straightforward.  At the time of commitment, it is possible to determine the lowest value at which the job would have been accepted (i.e., scheduled by the simulator). This critical value is the appropriate payment to guarantee truthfulness (see, e.g., \cite{HKMP05}), so it can be offered promptly.  

It is important to understand why $\Alg_{\tilde{T}C}$ may give incentive to misreport arrival times. Consider the single server case, take $\omega = 1/2$, and suppose there are two jobs $\type_1 = \langle v_1, D_1, a_1, d_1 \rangle = \langle 1, 1, 0, 8 \rangle$ and $\type_2 = \langle v_2, D_2, a_2, d_2 \rangle = \langle 10, 2, 0, 100 \rangle$. In this instance, job 1 would not be accepted: the simulator will process job $2$ throughout the interval $[0,4]$ (recall that demands are doubled in the simulation), blocking the execution of job $1$.  Since time $4$ is the virtual deadline of job $1$ (half of its execution window), the job will be rejected at that time. However, if job $1$ instead declared an arrival time of $4$, then the simulator would successfully complete the job by its virtual deadline of $6$, and the job would be accepted.

\subsection{Full Truthfulness}
\label{sec:Committed_Truthful_Full}

In this subsection, we explicitly construct a truthful, committed scheduling mechanism. The issue in the last example is that misreporting a later arrival time can lead to a later virtual deadline being used by the simulator. This ability to delay the virtual deadline can incentivize non-truthful reporting. We address this issue by imposing additional structure on the time intervals used for simulation. Given the reported job demand $D_j$ and execution window $[a_j, d_j]$, we determine a collection of subintervals of $[a_j, d_j]$ in which to run simulations.  If the simulator accepts the job in \emph{any} of these subintervals, we admit the job and process it in the subsequent interval; otherwise the job is rejected.  We will construct the subintervals in such a way that monotonicity is preserved: declaring a smaller execution window or a greater demand can lead only to less desirable simulation windows (i.e., subsets of the originals).

Truthfulness follows from the fact that the simulation parameters cannot be influenced beneficially by the reported arrival and departure times. The main technical challenge is to establish a competitive ratio bound for this modified solution; it turns out that the dual-fitting argument used to bound the competitive ratio of $\Alg_T$ in Section \ref{sec:NonCommitted_Truthfulness} can be modified to provide the necessary bounds. We end up with the following result.  A full proof, and a more formal description of the reduction, 
appears in Appendix \ref{sec:Appendix_Committed_Truthful}.


\begin{theorem}
\label{thm:Committed_FullTruthfulness}
There exist constants $c_0$ and $s_0$ such that, for any $s > s_0$, there exists a truthful, $(2 s / s_0)$-responsive scheduling algorithm $\Alg_{TC}$ such that:
\begin{eqnarray*}
\emph{\CompetitiveRatio}_{\Alg_{TC}}(s) & = & c_0 + \Theta\left(\frac{1}{\sqrt[3]{\ffrac{s}{s_0}}-1}\right) + \Theta\left(\frac{1}{(\sqrt[3]{\ffrac{s}{s_0}}-1)^3}\right).
\end{eqnarray*}
\end{theorem}

For the case of multiple identical servers, we obtain constants $c_0 = 94.248$ and $s_0 = 139.872$.  For the single server case, we obtain $c_0 = 9$ and $s_0 = 12$.




\section{Conclusion}
\label{sec:Conclusions}

This paper designs and analyzes truthful online scheduling mechanisms. 
Although the model studied herein is clearly a theoretical abstraction of the full complexity faced by scheduling of tasks in the cloud, we believe that the principles developed here can carry over to more complex settings.

%

The $\beta$-responsive mechanisms described in Section \ref{sec:Committed} and Section \ref{sec:Committed_Truthful_PublicArrivalTimes} actually satisfy two stronger properties. First, they satisfy an alternate responsiveness property: there exists a constant $\omega \in (0,1)$ such that the scheduler makes a commitment for each job after a $(1-\omega)$ fraction of the job execution window has passed, i.e., by time $d_j - \omega (d_j - a_j)$. Second, they satisfy no early processing, i.e., the mechanisms may process jobs only once they have committed to their completion. In contrast, the truthful scheduling mechanism from Section \ref{sec:Committed_Truthful_Full} does not necessarily satisfy the two properties. An interesting open question is whether there exists a (fully) truthful scheduling mechanism with constant competitive ratio that commits to scheduling each job before a constant fraction of its execution window has elapsed.

The most obvious problem left open by our work is to improve the constants in our results.  The mechanisms constructed for our most general results involve large constants that can potentially be improved. One particularly interesting question along these lines is whether one can obtain an approximation factor that approaches $1$ as the number of servers $C$ grows large.
An additional avenue of future work is to extend our results to more sophisticated scheduling problems.  One might investigate jobs with parallelism, or jobs made up of many interdependent tasks (see, e.g., \cite{spaa14}), or the impact of non-uniform machines or time-varying capacity, and so on.  The primary question is then to determine to what extent deadline slackness helps to construct constant-competitive mechanisms for variations of the online scheduling problem.

\bibliographystyle{acmsmall}
\bibliography{ec2015}

\begin{thebibliography}{}

\bibitem[\protect\citeauthoryear{Archer and \'{E}va Tardos}{Archer and \'{E}va
  Tardos}{2001}]{ArcherTardos01}
{\sc Archer, A.} {\sc and} {\sc \'{E}va Tardos}. 2001.
\newblock Truthful mechanisms for one-parameter agents.
\newblock In {\em FOCS}. 482--491.

\bibitem[\protect\citeauthoryear{Azure}{Azure}{2015}]{AzureMLPricing}
{\sc Azure}. 2015.
\newblock Azure machine learning pricing.
\newblock
  {\url{http://azure.microsoft.com/en-us/pricing/details/machine-learning/}}.

\bibitem[\protect\citeauthoryear{Bar-Noy, Canetti, Kutten, Mansour, and
  Schieber}{Bar-Noy et~al\mbox{.}}{1999}]{BCKMS99}
{\sc Bar-Noy, A.}, {\sc Canetti, R.}, {\sc Kutten, S.}, {\sc Mansour, Y.}, {\sc
  and} {\sc Schieber, B.} 1999.
\newblock Bandwidth allocation with preemption.
\newblock {\em SIAM J. Comput.\/}~{\em 28,\/}~5, 1806--1828.

\bibitem[\protect\citeauthoryear{Bod{\'\i}k, Menache, Naor, and
  Yaniv}{Bod{\'\i}k et~al\mbox{.}}{2014}]{spaa14}
{\sc Bod{\'\i}k, P.}, {\sc Menache, I.}, {\sc Naor, J.~S.}, {\sc and} {\sc
  Yaniv, J.} 2014.
\newblock Brief announcement: deadline-aware scheduling of big-data processing
  jobs.
\newblock In {\em SPAA}. 211--213.

\bibitem[\protect\citeauthoryear{Canetti and Irani}{Canetti and
  Irani}{1998}]{CI98}
{\sc Canetti, R.} {\sc and} {\sc Irani, S.} 1998.
\newblock Bounding the power of preemption in randomized scheduling.
\newblock {\em SIAM J. Comput.\/}~{\em 27,\/}~4, 993--1015.

\bibitem[\protect\citeauthoryear{Chan, Lam, and To}{Chan
  et~al\mbox{.}}{2005}]{CLT05}
{\sc Chan, H.}, {\sc Lam, T.~W.}, {\sc and} {\sc To, K.} 2005.
\newblock Nonmigratory online deadline scheduling on multiprocessors.
\newblock {\em {SIAM} Journal of Computing\/}~{\em 34,\/}~3, 669--682.

\bibitem[\protect\citeauthoryear{Curino, Difallah, Douglas, Krishnan,
  Ramakrishnan, and Rao}{Curino et~al\mbox{.}}{2014}]{rayon}
{\sc Curino, C.}, {\sc Difallah, D.~E.}, {\sc Douglas, C.}, {\sc Krishnan, S.},
  {\sc Ramakrishnan, R.}, {\sc and} {\sc Rao, S.} 2014.
\newblock Reservation-based scheduling: If you're late don't blame us!
\newblock In {\em Proceedings of the ACM Symposium on Cloud Computing}. ACM,
  1--14.

\bibitem[\protect\citeauthoryear{DasGupta and Palis}{DasGupta and
  Palis}{2000}]{DP00}
{\sc DasGupta, B.} {\sc and} {\sc Palis, M.~A.} 2000.
\newblock Online real-time preemptive scheduling of jobs with deadlines.
\newblock In {\em APPROX}. 96--107.

\bibitem[\protect\citeauthoryear{Ferguson, Bodik, Kandula, Boutin, and
  Fonseca}{Ferguson et~al\mbox{.}}{2012}]{jockey}
{\sc Ferguson, A.}, {\sc Bodik, P.}, {\sc Kandula, S.}, {\sc Boutin, E.}, {\sc
  and} {\sc Fonseca, R.} 2012.
\newblock Jockey: guaranteed job latency in data parallel clusters.
\newblock In {\em Proceedings of the 7th ACM european conference on Computer
  Systems}. ACM, 99--112.

\bibitem[\protect\citeauthoryear{Funk}{Funk}{2004}]{Funk04}
{\sc Funk, S.~H.} 2004.
\newblock {EDF} scheduling on heterogeneous multiprocessors.
\newblock Ph.D. thesis, University of North Carolina.

\bibitem[\protect\citeauthoryear{Garay, Naor, Yener, and Zhao}{Garay
  et~al\mbox{.}}{2002}]{GNYZ02}
{\sc Garay, J.~A.}, {\sc Naor, J.}, {\sc Yener, B.}, {\sc and} {\sc Zhao, P.}
  2002.
\newblock On-line admission control and packet scheduling with interleaving.
\newblock In {\em INFOCOM}.

\bibitem[\protect\citeauthoryear{Hajiaghayi, Kleinberg, Mahdian, and
  Parkes}{Hajiaghayi et~al\mbox{.}}{2005}]{HKMP05}
{\sc Hajiaghayi, M.~T.}, {\sc Kleinberg, R.}, {\sc Mahdian, M.}, {\sc and} {\sc
  Parkes, D.~C.} 2005.
\newblock Online auctions with re-usable goods.
\newblock 165--174.

\bibitem[\protect\citeauthoryear{Hong and Leung}{Hong and Leung}{1989}]{HL89}
{\sc Hong, K.~S.} {\sc and} {\sc Leung, J.~Y.} 1989.
\newblock Preemptive scheduling with release times and deadlines.
\newblock {\em Real-Time Systems\/}~{\em 1,\/}~3, 265--281.

\bibitem[\protect\citeauthoryear{Jain, Menache, Naor, and Yaniv}{Jain
  et~al\mbox{.}}{2011}]{JMNY11}
{\sc Jain, N.}, {\sc Menache, I.}, {\sc Naor, J.}, {\sc and} {\sc Yaniv, J.}
  2011.
\newblock A truthful mechanism for value-based scheduling in cloud computing.
\newblock In {\em SAGT}. 178--189.

\bibitem[\protect\citeauthoryear{Jain, Menache, Naor, and Yaniv}{Jain
  et~al\mbox{.}}{2012}]{JMNY12}
{\sc Jain, N.}, {\sc Menache, I.}, {\sc Naor, J.}, {\sc and} {\sc Yaniv, J.}
  2012.
\newblock Near-optimal scheduling mechanisms for deadline-sensitive jobs in
  large computing clusters.
\newblock In {\em SPAA}. 255--266.

\bibitem[\protect\citeauthoryear{Koren and Shasha}{Koren and
  Shasha}{1992}]{KS92}
{\sc Koren, G.} {\sc and} {\sc Shasha, D.} 1992.
\newblock D$^{\mbox{over}}$; an optimal on-line scheduling algorithm for
  overloaded real-time systems.
\newblock In {\em RTSS}. IEEE Computer Society, 290--299.

\bibitem[\protect\citeauthoryear{Koren and Shasha}{Koren and
  Shasha}{1994}]{KS94}
{\sc Koren, G.} {\sc and} {\sc Shasha, D.} 1994.
\newblock Moca: A multiprocessor on-line competitive algorithm for real-time
  system scheduling.
\newblock {\em Theor. Comput. Sci.\/}~{\em 128,\/}~1{\&}2, 75--97.

\bibitem[\protect\citeauthoryear{Lavi and Swamy}{Lavi and
  Swamy}{2007}]{LaviSwamy07}
{\sc Lavi, R.} {\sc and} {\sc Swamy, C.} 2007.
\newblock Truthful mechanism design for multi-dimensional scheduling via cycle
  monotonicity.
\newblock In {\em EC}.

\bibitem[\protect\citeauthoryear{Lucier, Menache, Naor, and Yaniv}{Lucier
  et~al\mbox{.}}{2013}]{LMNY13}
{\sc Lucier, B.}, {\sc Menache, I.}, {\sc Naor, J.}, {\sc and} {\sc Yaniv, J.}
  2013.
\newblock Efficient online scheduling for deadline-sensitive jobs.
\newblock In {\em SPAA}. 305--314.

\bibitem[\protect\citeauthoryear{Phillips, Stein, Torng, and Wein}{Phillips
  et~al\mbox{.}}{1997}]{PSTW97}
{\sc Phillips, C.~A.}, {\sc Stein, C.}, {\sc Torng, E.}, {\sc and} {\sc Wein,
  J.} 1997.
\newblock Optimal time-critical scheduling via resource augmentation.
\newblock In {\em Proceedings of the Twenty-Ninth Annual {ACM} Symposium on the
  Theory of Computing, El Paso, Texas, USA, May 4-6, 1997}. 140--149.

\bibitem[\protect\citeauthoryear{Porter}{Porter}{2004}]{Porter04}
{\sc Porter, R.} 2004.
\newblock Mechanism design for online real-time scheduling.
\newblock In {\em In Proc. ACM Conf. on Electronic Commerce (EC)}. ACM Press,
  61--70.

\bibitem[\protect\citeauthoryear{Vazirani}{Vazirani}{2001}]{Vazirani01}
{\sc Vazirani, V.~V.} 2001.
\newblock {\em Approximation algorithms}.
\newblock Springer.

\end{thebibliography}

\newpage

\appendix{\textbf{\large Appendices}}






%

\medskip



\section{An Alternative Notion of Promptness}
\label{sec:Appendix_Committed_Definition}

Recall the definition of $\beta$-responsiveness: a scheduling mechanism is $\beta$-responsive (for $\beta \geq 0$) if, for every job $j$, by time $d_j - \beta \cdot D_j$ it either (a) rejects the job or, (b) guarantees that the job will be completed by its deadline and specifies the required payment.  
In this section we discuss a different, equally natural notion of responsiveness.

Given $\omega \in [0,1]$, we could ask for a scheduling mechanism to make the choice of whether to accept or reject each job $j$ by time $d_j - \omega (d_j - a_j)$.  That is, a decision must be reached for each job when a $(1 - \omega)$ fraction of its execution window has elapsed.  The case $\omega = 1$ corresponds to allocation decisions being made upon arrival, and $\omega = 0$ is equivalent to no commitment.  

We note that the mechanisms constructed in Section \ref{sec:Committed} (non-truthful) and Section \ref{sec:Committed_Truthful_PublicArrivalTimes} (truthful when arrival times are public) actually satisfy this alternative notion of responsiveness for a constant $\omega$, in addition to being $\beta$-responsive.  Indeed, for these mechanisms, $\beta$-responsiveness actually follows as a corollary of this alternative form of multiplicative responsiveness, combined with the slackness condition.  However, the truthful and $\beta$-responsive mechanism from Section \ref{sec:Committed_Truthful_Full} does not satisfy multiplicative responsiveness for any constant $\omega$.  We leave open the problem of designing a fully truthful scheduler that makes commitments before a constant faction of each job's execution window has passed.


\section{Truthful Non-Committed Scheduling}
\label{sec:Appendix_NonCommitted_Truthful}


\subsection{Non-Truthfulness of \cite{LMNY13}}
\label{sec:Appendix_NonCommitted_Truthful_SPAA2013}

Recall that the non-committed algorithm by \cite{LMNY13} is based on the following two properties. First, a running job $j$ can only be preempted by a job $j'$ satisfying $\vd_{j'} > \thres \vd_j$ for some parameter $\gamma > 1$. Second, if a job $j$ is not allocated by time $d_j - \gap D_j$ for some $\gap \ge 1$, it is not allocated at all. In the following, we prove that such a scheduler is not truthful.

Assume the system consists of a single server. Consider four job types: $A,B,C$ and $D$. Assume $\vd_A = 1$, $\vd_B = \thres$, $\vd_C = \thres ^2$ and $\vd_D = \infty$. Specifically, type $B$ jobs cannot preempt type $A$ jobs; type $C$ jobs cannot preempt type $B$ jobs; however, type $C$ jobs can preempt type $A$ jobs. We use type $D$ jobs to maintain the server busy when needed.
Our input consists of one type $A$ job (which we simply refer to as $A$), one type $B$ job (referred as B) and $s$ type $C$ jobs. We construct an instance such that $A$ is not completed due to type $C$ jobs. However, by decreasing the value of $A$, job $B$ blocks the type $C$ jobs from running. This allows $A$ to complete.

Set $D_A = 2\gap$ and $d_A = s+\gap$. Set $a_B = 0.5 \gap$ and $D_B = s - 0.5\gap$. Finally, set $a_C = \gap$, $d_C = s + \gap$ and $D_C = 1$. Assume all other parameters are set such every job $j$ satisfies $v_j = \vd_j D_j$ and $d_j - a_j = sD_j$. Type $D$ jobs are set such that the server is busy until time $t=0$.
\\

\noindent\textbf{Case 1 - $\vd_A = 1$.}
\begin{table}[h]{
\begin{tabular}{ll}
$t<0$ & The algorithm processes type $D$ jobs. \\
$t=0$ & The algorithm begins to process $A$. \\
$t=0.5\gap$ & Job $B$ arrives. The algorithm decides not to preempt $A$. \\
$t=\gap$ & All type $C$ jobs arrive. Job $A$ is preempted.\\
              & The algorithm processes $s$ type $C$ jobs until time $t=s+\gap$. \\
$t=s+\gap$ & The type $C$ jobs are all processed, but $A$ is not completed by its deadline.

\end{tabular}}
\end{table}

\noindent\textbf{Case 2 - $\vd_A < 1$.}
\begin{table}[h]{
\begin{tabular}{ll}
$t<0$ & The algorithm processes type $D$ jobs. \\
$t=0$ & The algorithm begins to process $A$. \\
$t=0.5\gap$ & Job $B$ arrives. The algorithm preempts $A$ and begins to process $B$. \\
$t=\gap$ & All type $C$ jobs arrive. Job $B$ is not preempted.\\
$t=s$ & The algorithm completes $B$. The type $C$ were not allocated by $d_C - \gap D_C = s$.\\
& Hence, all type $C$ jobs are rejected. The algorithm resumes processing $A$. \\
$t=s + \gap$ & The algorithm completes $A$ by its deadline.

\end{tabular}}
\end{table}


\subsection{Single Server}
\label{sec:Appendix_NonCommitted_Truthful_SingleServer}

We prove the truthfulness of the non-committed single server algorithm.
\\

\begin{proofof}{Theorem \ref{thm:NonCommitted_Truthful_SingleServer_Truthfulness}}
By Theorem \ref{thm:truth}, it suffices to show that $\Alg_T$ is monotone.  Consider some job $j$. Throughout the proof, we fix the types $\type_{-j}$ of all jobs beside $j$. To ease exposition, we drop $\type_{-j}$ from our notation.  Write $\tau_j = \langle v_j, D_j, a_j, d_j \rangle$ for the \emph{true} type of job $j$.  Suppose that $\Alg_T(\tau_j, \type_{-j})$ completes job $j$.  We must show that $\Alg_T$ will still complete job $j$ under a reported type of $\tau_j'$, where $\tau_j' \succ \tau_j$.  Since one can modify each component of a reported type in sequence, it will suffice to establish monotonicity with respect to each coordinate independently.

\medskip
\noindent
\textbf{Step 1: Value Monotonicity. }
Let us first establish value monotonicity.  Consider some $v_j' > v_j$ and assume $j$ is completed when reporting $v_j$.
Let $\vd'_j = \ffrac{v'_j}{D_j}$ be the value-density and let $\ell'_j = \lfloor \log_\thres(\vd'_j) \rfloor$ be the class of job $j$ when reporting $v'_j$. Let $st'(y_j)$ be the starting time of $j$ when reporting $v'_j$. Similarly, denote $\vd_j$, $\ell_j$ and $st(y_j)$ with respect to $v_j$.
Notice that if $\ell_j = \ell'_j$, then the behavior of $\Alg_T$ is identical regardless of value reported; hence $j$ is completed.  We therefore assume $\ell'_j > \ell_j$.  Note then that $st'(y_j) \le st(y_j)$, since if $j$ does not start by time $st(y_j)$ under reported value $v_j'$, then it will start at that time since it has only a higher class.

%
Now, consider the case where $j$ reports a value of $v_j'$. Observe $J^P(t)$ at time $t=st'(y_j)$. Claims \ref{thm:NonCommitted_Truthful_SingleServer_Claim2} and \ref{thm:NonCommitted_Truthful_SingleServer_Claim4} imply that no existing job can preempt job $j$. Specifically, any job that can run instead of $j$ during the interval $[st'(y_j),d_j]$ must have arrived after time $st'(y_j)$. If $j$ did not complete, then during this interval the algorithm processed higher priority jobs during more than $d_j - st'(y_j) - D_j$ time units. These jobs would also be preferred by $\Alg_T$ when $j$ reports a lower value of $v_j$. Hence, $j$ could not have been completed when the value $v_j$ was reported, a contradiction.

\medskip

\noindent
\textbf{Step 2: Monotonicity of Other Properties. }
Next consider misreporting a demand $D'_j \leq D_j$, and suppose the job completes under report $D_j$.  Then the job's value density is higher under report $D'_j$, and its latest possible starting time is increased to $d_j - \gap D'_j$. This only extends the possibilities of $j$ being completed, and hence job $j$ would be completed under report $D_j'$ as well.  Following similar arguments, a later deadline $d'_j$ instead of $d_j$ only increases the latest possible start time of $j$, and increases the time slots in which the job can be completed, which again can only increase the allocation to $j$. Finally, we show that a later arrival time $a'_j \le a_j$ cannot be detrimental.  We assume that $j$ is completed when the job is submitted at time $a_j$.  It remains to prove that $j$ is completed when submitting at time $a'_j$. This follows in a similar fashion to our argument for value monotonicity. If when reporting $a'_j$ the job is not processed until $a_j$, then both executions of $\Alg_T$ are identical, hence $j$ is completed. Otherwise, $j$ necessarily begins earlier than when $a_j$ is reported, and again job $j$ is completed. We reach the same conclusion as before.
%

Since $\Alg_T$ satisfies all required monotonicity conditions, we conclude that $\Alg_T$ is truthful.
\end{proofof}

We now bound the competitive ratio of the truthful non-committed algorithm for a single server.
\\

\begin{proofof}{Theorem \ref{thm:NonCommitted_Truthful_SingleServer}}
We bound the competitive ratio of $\Alg_T$ for the single server case.
Our proof strongly relies on the original analysis of the non-committed scheduling algorithm by \cite{LMNY13}. Consider an execution of $\Alg_T$ on an instance $\type$.
Recall that Claim \ref{thm:NonCommitted_Truthful_SingleServer_Claim2} states that at each time $t$, no job in $J^P(t)$ or $J^E(t)$ can have a value density larger than $\thres \vd_{\Alg_T}(t)$. Furthermore, the algorithm does not allocate resources to any job $j$ that has not been allocated by time $d_j - \mu D_j$. For schedulers satisfying these two properties, \cite{LMNY13} proved that there exists a feasible solution for the dual program corresponding to $\type$, with a total dual cost of:
\begin{eqnarray}
\label{eq:NonCommitted_Truthful_SingleServer_CompetitiveRatio1}
v(\Alg_T(\tau)) \, + \, \thres \cdot \frac{s}{s-\gap} \cdot \intop_{0}^{\infty} \vd_{\Alg_T}(t)dt.
\end{eqnarray}

\noindent
It remains to bound the integral.
Let $\TimeFull$ denote the set of times during which completed jobs were processed, and denote by $\TimePartial$ the remaining times. Notice that all jobs processed during $\TimePartial$ are partially processed jobs. Hence, the integral can be written as $v(\Alg_T(\tau)) + \int_{\TimePartial} \vd_{\Alg_T}(t) dt$. The latter expression represents the total value corresponding to partial work lost from not completing jobs. That is, if the algorithm processed half of some job $j$, then it lost a partial value of $0.5 v_j$.

Several useful properties of the original non-truthful algorithm are preserved in our truthful variant. We prove here that $\Alg_T$ preserves these properties and show how they can be used to bound the lost partial value. Consider a partially processed job $j$. Let $j'$ be any job other than $j$ running during some time $t\in[st(y_j),d_j]$. We claim that $st(y_j) < st(y_{j'})$. Assume the contrary. Since $j$ starts processing at time $st(y_j)$, this implies that $j \succ j' $. However, we know that $j$ has never been completed. By Claim \ref{thm:NonCommitted_Truthful_SingleServer_Claim2}, it is impossible that $j'$ was processed at time $t$. Therefore, the claim holds. Moreover, it follows that $j' \succ j$, since $j'$ started being processed after time $st(y_j)$.

Consider again the interval $[st(y_j),d_j]$. Notice that the length of the interval is at least $\gap D_j$. Our previous claims imply that during this interval the algorithm processed jobs that belong to higher classes than $j$ for at least $(\gap - 1)D_j$ of the time. This translates to a value of at least $(\gap -1)v_j$, since the value density of these jobs are at least $\vd_j$. Intuitively, one would wish that this value could account for the loss of $j$. However, jobs processed inside the interval $[st(y_j),d_j]$ have not necessarily been completed. This calls for a more rigorous analysis. \cite{LMNY13} introduced a complex charging argument to bound the lost partial value. By slightly modifying their proof\footnote{Our analysis differs since the truthful algorithm preempts jobs according to class. Consider two jobs $j,j'$ that belong to classes $\ell,\ell'$, respectively. Notice that if $\ell' - \ell = i > 0$ then $\vd_{j'} \le \thres^{i-1} \vd_j$. This bound is weaker by a factor of $\thres$ compared to an equivalent bound obtained by \cite{LMNY13}, which increases the bound on the lost partial value by $\thres$.}, we can obtain the following bound:
\begin{eqnarray}
\label{eq:NonCommitted_Truthful_SingleServer_CompetitiveRatio2}
\intop_{0}^{\infty} \vd_{\Alg_T}(t)dt
  & \,\,\le\,\, &
v(\Alg_T(\type)) \cdot \left[ 1 + \frac{\thres}{(\thres - 1)(\gap - 1) - 1} \right]
\end{eqnarray}
By combining \eqref{eq:NonCommitted_Truthful_SingleServer_CompetitiveRatio1} and \eqref{eq:NonCommitted_Truthful_SingleServer_CompetitiveRatio2} we can then apply the dual fitting theorem (Theorem \ref{thm:DualFitting}) and get:
\begin{eqnarray}
\label{eq:NonCommitted_Truthful_SingleServer_CompetitiveRatio3}
\CompetitiveRatio_{\Alg_T}(s)
  & \le &
1 + \thres \cdot \frac{s}{s-\gap} \cdot \left[ 1 + \frac{\thres}{(\thres - 1)(\gap - 1) - 1} \right].
\end{eqnarray}
For every $\gap$, the above bound is optimized for a unique value $\thres^*(\gap) = \frac{\sqrt{\gap}}{\sqrt{\gap} - 1}$. By choosing $\gap \approx s^{\ffrac{2}{3}}$ we obtain the bound stated in the theorem.
\end{proofof}


\subsection{Multiple Servers}
\label{sec:Appendix_NonCommitted_Truthful_MultipleServers}

We extend $\Alg_T$ to accommodate multiple servers.
In the multiple server variant, which we also denote by $\Alg_T$, each server runs a local copy of the single server algorithm. Specifically, if a job $j$ has not been executed on a server $i$ during the interval $[a_j, d_j - \gap D_j]$, then the algorithm prevents it from running on server $i$. The detailed implementation of the algorithm is given fully in Algorithm \ref{alg:Truthful_NonCommitted_MultipleServers}. The algorithm follows the general approach described here in an efficient manner. Upon arrival of a job $j$ at time $t$, we only invoke the class preemption rule on server $i_{\textrm{min}}(t)$, which is the server running the job belonging to the lowest class (unused servers run idle jobs of class $-\infty$). Ties are broken in favor of the job with the later start time in the system. This is crucial for proving truthfulness. Notice also that it suffices to invoke the class preemption rule of server $i_{\textrm{min}}(t)$: if job $j$ is rejected, it would be rejected by the class preemption rule of any other server. When job $j$ completes on server $i$, we first load the job with maximal value-density out of the jobs preempted from server $i$, and then invoke the class preemption rule.
Notice that the class preemption rule allows preempted jobs to migrate between servers. That is, a job preempted from server $i$ might start executing on a different server at time $t$, provided that $t \le d_j - \gap D_j$.


\begin{algorithm}
\label{alg:Truthful_NonCommitted_MultipleServers}
\SetKw{KwEvent}{Event:}
\DontPrintSemicolon
\caption{Truthful Non-Committed Algorithm $\Alg_T$ for Multiple Servers}


$\forall t, \,\,\,\,\, J_i^P(t) = \left\{ \,j\in\JobInput \mid j \textrm{ partially processed on server } i \textrm{ at time } t \,\wedge\, t\in [a_j,d_j]  \right\}$.\;
$\,\,\,\,\,\,\,\,\,\,\,\,\, J_i^E(t) = \left\{ \,j\in\JobInput \mid j \textrm{ unallocated on server } i \textrm{ at time } t \,\wedge\, t\in [a_j,d_j - \gap D_j]  \right\}$.\;
$\,\,\,\,\,\,\,\,\,\,\,\,\, J_{\Alg}(t) = \left\{ \,j_\Alg^i(t) \mid 1 \le i \le C \right\}$. $\qquad\qquad\qquad\,\,\,\,\,\,\,\,\,$ (jobs executing at time $t$) \;
\;

\KwEvent On arrival of job $j$ at time $t=a_j$:\;
$\,\,\,$ 1. call ClassPreemptionRule($i_{\textrm{min}}(t),t)$, where:\; $\,\,\,\,\,\,\,\,\,$ $i_{\textrm{min}}(t) = \arg\!\min \{ \, \lfloor \log_{\thres} \vd_{\Alg}^i(t) \rfloor \, \mid \, 1 \le i \le C \, \}$ $\,\,\,\,\,\,\,\,\,$ (ties broken according to later start time)\;
\;

\KwEvent On completion of job $j$ on server $i$ at time $t$:\;
$\,\,\,$ 1. resume execution of job $j' = \arg\max \{ \vd_{j'} \mid j' \in J^P_i(t) \}$.\;
$\,\,\,$ 2. call ClassPreemptionRule($i$,$t$).\;
$\,\,\,$ 3. delay the output response of $j$ until time $d_j$.\;
\;

\textbf{ClassPreemptionRule($i$,$t$):}\;
$\,\,\,$ 1. $j \,\,\, \leftarrow$ job currently being processed on server $i$.\;
$\,\,\,$ 2. $j^* \leftarrow \arg\!\max \left\{ \vd_{j^*} \mid j \in J_i^E(t) \setminus J_\Alg(t) \right\}$ $\,\,\,\,\,\,\,\,\,\,\,\,\,\,\,\,\,\,\,\,\,\,$ {\small (ties broken by earlier start time)}\;
$\,\,\,$ 3. if $\left( j^* \succ j \right)$\;
$\,\,\,\,\,\,\,\,\,\,\,\,\,\,\,\,\,\,$ 3.1. preempt $j$ and run $j^*$.\;

\end{algorithm}


We argue that the proposed mechanism is truthful.
Note that claims \ref{thm:NonCommitted_Truthful_SingleServer_Claim2} and \ref{thm:NonCommitted_Truthful_SingleServer_Claim4} apply on each server separately. However, this is insufficient for proving truthfulness. Instead, we prove the following useful claim.
Recall that $y_j^i(t)$ indicates whether job $j$ was allocated on server $i$ at time $t$. Define the \emph{starting point} $st(y_j) = \min \big\{ \{ t \mid y_j^i(t) = 1 \} \cup \{ \infty \}\big\}$ of job $j$ as the first point in time at which $j$ is allocated. If no such $t$ exists, $st(y_j)=\infty$.

\begin{claim}
Let $j = j_{\Alg_T}^i(t)$ be the job processed on server $i$ at time $t$ by $\Alg_T$. Let $j'$ be any job not running at time $t$, and assume $j'$ is either an allocated job such that $t\in[a_{j'},d_{j'}]$ or an unallocated job such that $t\in [a_{j'},d_{j'} - \gap D_j]$. Let $\class_\ell,\class_{\ell'}$ denote the classes of $j,j'$, respectively. Then, either $\ell > \ell'$ or $\ell = \ell' \,\wedge\, st(y_j) < st(y_{j'})$.
\end{claim}

\begin{proof}
Since each server runs a local copy of the single server algorithm, Claim \ref{thm:NonCommitted_Truthful_SingleServer_Claim2} implies that $j' \not\succ j$, therefore $\ell \ge \ell'$. It remains to prove that if $\ell = \ell'$ then $st(y_j) < st(y_{j'})$. Assume towards contradiction that $st(y_{j'}) < st(y_j)$ (equality is impossible, since we assume $j$ is running at time $t$ and $j'$ is not). This implies that the algorithm always prioritizes $j'$ over $j$. Notice that at time $st(y_j)$ job $j'$ must be running; otherwise, the algorithm would have not started processing job $j$. Therefore, at time $st(y_j)$ both jobs are running, and at time $t$ only job $j$ is running. We show that this scenario is impossible. Notice that $j'$ cannot be preempted while $j$ is running, since the algorithm would choose to preempt $j$ instead. Hence, sometime during the interval $[st(y_j),d_j]$ both jobs were preempted and $j$ resumed execution. This is impossible, since $j'$ would have been resumed instead of $j$. We reach a contradiction. Therefore, the claim holds.
\end{proof}

The claim implies that at every time $t$ the algorithm is processing the $C$ top available jobs, where the jobs are ordered first by their class (high to low), and in case of equality ordered by their start times (low to high). This observation is essential for proving truthfulness.

\begin{theorem}
\label{thm:NonCommitted_Truthful_MultipleServers_Truthfulness}
The algorithm $\Alg_T$ for multiple servers is truthful.
\end{theorem}

\begin{proof}
The proof follows directly from the equivalent single server proof.
Consider some job $j$ and two value-densities $\vd'_j \le \vd''_j$. Let $\ell', \ell'$ denote the corresponding classes and let $st'(y_j), st''(y_j)$ denote the corresponding start times. Notice that $\ell' \le \ell''$. We prove that also $st''(y_j) \le st'(y_j)$. Consider the case where $j$ has a value-density of $\vd''_j$. If $j$ is processed before time $st'(y_j)$, the claim holds. Otherwise, the behavior of the algorithm up to time $st'(y_j)$ is identical in both cases. Since now $j$ has a higher value density, it will also begin processing.
We conclude that increasing the value-density only increases the priority of $j$, with respect to the algorithm $\Alg_T$. Therefore, we can repeat the arguments that lead to prove the truthfulness of the single server algorithm.
\end{proof}

We conclude by proving the bound on the competitive ratio stated in Theorem \ref{thm:NonCommitted_Truthful_MultipleServers}.
\\

\begin{proofof}{Theorem \ref{thm:NonCommitted_Truthful_MultipleServers}}
Similar to the single server case, we can show that for every time $t$ a running job $j$ and a pending job $j'$ satisfy $\vd_{j'} \le \thres \vd_j$. Furthermore, each server does not begin processing any job $j$ after time $d_j - \gap D_j$. \cite{LMNY13} proved that in this case, there exists a feasible solution for the dual program corresponding to $\type$, with a total dual cost of:
\begin{eqnarray}
\label{eq:NonCommitted_Truthful_MultipleServers_CompetitiveRatio1}
\bigg[ 1 + \thres \cdot \frac{s}{s-\gap} \bigg] \cdot \bigg[ v(\Alg_T(\type)) + \sum_{i=1}^C \intop_{0}^{\infty} \vd_{\Alg_T}^i(t)dt \bigg].
\end{eqnarray}
We can bound the integral $\intop_{0}^{\infty} \vd_{\Alg_T}^i(t)dt$ for each server $i$ individually, as done for the single server case. Summing over all servers, we get the following bound on the competitive ratio of $\Alg_T$:
\begin{eqnarray}
\label{eq:NonCommitted_Truthful_MultipleServers_CompetitiveRatio2}
\CompetitiveRatio(\Alg_T)
  & \le &
\bigg[ 1 + \thres \cdot \frac{s}{s-\gap} \bigg] \cdot \bigg[ 1 + \frac{\thres}{(\thres - 1)(\gap - 1) - 1} \bigg].
\end{eqnarray}
By setting $\thres = \frac{\sqrt{\gap}}{\sqrt{\gap} - 1}$ and $\gap \approx s^{\ffrac{2}{3}}$ we obtain the bound stated in the theorem.
\end{proofof}


\section{Lower Bound on Committed Scheduling}
\label{sec:Appendix_Committed_LowerBound}

In the following section we prove Theorem \ref{thm:Committed_LowerBound}.
We first prove that no single server committed scheduler can provide any constant competitive ratio for $s<4$. We then generalize our bound for $C\le 3$ servers, and prove an impossibility result for $s < \ffrac{4}{C}$.

\begin{theorem}
\label{thm:Committed_LowerBound_SingleServer}
In the single server model, any online algorithm that commits to jobs on admission has an unbounded competitive ratio for $s<4$.
\end{theorem}

To prove Theorem \ref{thm:Committed_LowerBound_SingleServer}, we describe the following adversarial strategy.
The adversary sets the value of each arriving job to be significantly larger than the sum of all previous jobs, and waits for the job to be accepted before submitting a new job. This forces the algorithm to admit all arriving jobs, otherwise the algorithm would not maintain a constant competitive ratio. In addition, all jobs share the same deadline.
We first make a simplifying assumption on the scheduling algorithm, which we later relax. Given that all deadlines are identical, it is natural to assume that the scheduling algorithm does not admit a job before completing all previous commitments. We call such an algorithm \emph{natural}.


\begin{lemma}
\label{thm:Committed_LowerBound_SingleServer_Natural}
In the single server model, any online \emph{natural} algorithm that commits to jobs on admission has an unbounded competitive ratio for $s<4$.
\end{lemma}

\begin{proof}
First note that in order to prove the lower bound, it is enough to consider work preserving algorithms. An algorithm is considered \emph{work preserving} if the algorithm does not remain idle if it has unmet commitments. We can assume this since every algorithm can be transformed into a work preserving algorithm and perform at least as good as the original algorithm for any input.

Assume towards contradiction that there is a natural algorithm with a bounded competitive ratio of $c \ge 1$. Denote by $j_1,j_2,\dots$ the jobs submitted in order of their submission. We construct an adversarial strategy subject to the following invariants.

\begin{invariant}
\label{invariant:inv1}
Every arriving job $j$ has a deadline $d_j = s$ and a demand $D_j = \frac{s - a_j}{s}$, which is the largest possible demand, with respect to the slackness constraint.
\end{invariant}

\begin{invariant}
\label{invariant:inv2refined}
A new job arrives immediately when the previous job is admitted. Formally, let $t_n$ denote the admission time of job $j_n$. Then, $a_{n+1} = t_n$.
\end{invariant}

\noindent
Recall that each job is associated with a type $\type_j = \big\langle v_j, D_j, a_j, d_j \big\rangle$. The first job $j_1$ has a type $\langle 1, 1, 0, s \rangle$.
As long as the algorithm does not accept the job, the adversary does not submit any additional jobs, as stated in Invariant \ref{invariant:inv2refined}.
Eventually, the algorithm must accept job $j_1$, otherwise the competitive ratio will be unbound.
At time $t_1$, the adversary submits the next job $j_2$ with a type $\langle c+1, \frac{s-t_1}{s}, t_1, s \rangle$. Specifically, the value of $j_2$ is significantly higher than $j_1$, and the demand is set according to \ref{invariant:inv1}.
Job $j_2$ must be accepted, to maintain the guaranteed competitive ratio. We now submit job $j_3$ and so forth.

In general, the demand of each job $j_n$ is set as $D_n = \frac{s - t_{n-1}}{s}$, in accordance to Invariant \ref{invariant:inv1}. The value of each job $j_n$ is set as $v_n = (c+1)^{n-1}$. Note that this value is at least $c$ times larger than the sum of all previous job values.
The adversary continues this strategy until the algorithm is forced to commit a job which cannot be completed by its deadline among with its previous commitments.

Define $\ell_{n+1}$ as the time between the completion of job $j_n$ and the admission of job $j_{n+1}$. Formally, we can write $\ell_{n+1} = t_{n+1} - (t_n + D_n)$, since the algorithm executes $j_n$ starting time $t_n$ without interruption. Note that for natural algorithms, $\ell_{n+1} \ge 0$, hence $t_{n+1} \ge t_n + D_n$.
Let $w(t)$ denote the total remaining unprocessed demand of admitted jobs at time $t$, that is, the total time needed for the scheduler to meet all remaining commitments. Correspondingly, let $f(t) = (s-t) - w(t)$ denote the available free time before the common deadline $s$. By the assumption that the scheduling algorithm is natural, it holds that $w(t_n) = D_n$ and therefore $f(t_n) = (s-t_n) - D_n$.
By combining the last equation with the definition of $D_j$ and the bound on $t_n$, we get that:
\begin{equation}
\label{eqn:mainXn}
s - t_{n+1} \,\,\,\,\le\,\,\,\, f(t_n) \,\,\,\,=\,\,\,\, s - t_n - D_n  \,\,\,\,=\,\,\,\, s - t_n - \frac{s-t_{n-1}}{s}
\end{equation}
Define $\Delta_n = s - t_n$ for every $n$. Intuitively, $\Delta_n$ represents the free time immediately before the admission of job $j_n$.
Equation \eqref{eqn:mainXn} can be written as:
\begin{eqnarray}
\label{eq:mainDeltan}
\Delta_{n+1} & \le & \Delta_n - \frac{\Delta_n-1}{s}
\end{eqnarray}
Next, we prove that if $s<4$ then $\Delta_n$ becomes negative, and therefore so does $f(t_n)$ by definition.
Assume towards contradiction that for every $n$ we have $\Delta_n > 0$.
Define $y_n = \frac{\Delta_{n-1}}{\Delta_n}$.
Notice that $y_n > 1$ for every $n$, since $\Delta_n$ is monotonically decreasing for $n$ by definition. By dividing \eqref{eq:mainDeltan} by $\Delta_n$ and using $y_n$, we get that $\frac{1}{y_n} \le 1 - \frac{y_{n-1}}{s}$. By using the known inequality $1 - \frac{1}{\alpha} \le \frac{\alpha}{4}$ for every $\alpha \ge 0$, we get that $\frac{1}{y_n} \le \frac{s}{4 y_{n-1}}$, or alternatively, $\frac{y_n}{y_{n-1}} > \frac{4}{s} > 1$ (since $s<4$). Hence, when $n \rightarrow \infty$ we get that $y_n \rightarrow \infty$.
However, by \eqref{eq:mainDeltan} we have: $0 \le \Delta_{n+1} \le \Delta_n - \frac{\Delta_{n-1}}{s}$ and therefore $\frac{\Delta_{n-1}}{\Delta_n} = y_n \le s$, which is a contradiction.
\end{proof}


Until now, we have only considered natural algorithms. Notice that for general scheduling algorithms, some of the values $\ell_n$ might be negative. To overcome this difficulty, we modify the strategy of the adversary.
\\

\begin{proofof}{Theorem \ref{thm:Committed_LowerBound_SingleServer}}
Assume towards contradiction there is a (general) algorithm that guarantees a competitive ratio of $c \ge 1$. As in the proof of Lemma \ref{thm:Committed_LowerBound_SingleServer_Natural}, it is enough to consider work preserving algorithms. We slightly modify the adversary described in Lemma \ref{thm:Committed_LowerBound_SingleServer_Natural} to handle cases where $\ell_n < 0$.
Let $n$ be the first job for which $\ell_n < 0$. This means that the algorithm admits $j_n$ before all previous commitments have been met.
Recall that $j_n$ arrives at time $t_{n-1}$, when job $j_{n-1}$ is admitted. Hence, job $j_n$ is admitted sometime during the execution of $j_{n-1}$, since by our choice of $n$, all previous jobs $j_1 \dots j_{n-2}$ have been completed.
The adversary does the following. Instead of submitting job $j_{n+1}$ immediately at time $t_n$ (as stated in Invariant \ref{invariant:inv2refined}), the adversary waits $|\ell_n|$ time before submitting $j_{n+1}$. Let $a'_{n+1} = t_n + \max\{0, -\ell_n\}$ denote the new arrival time of job $j_{n+1}$. Notice that $a'_{n+1}$ corresponds to the arrival time $a_{n+1}$ of $j_{n+1}$ if $\ell_n$ would have been $0$. We claim that by waiting $|\ell_n|$ time, the adversary sees the same setting at time $a'_{n+1}$ as he would for a natural algorithm with all previous values $\ell_1,\dots,\ell_{n-1} \ge 0$ and $\ell_n = 0$.
This follows since we assume the algorithm is work preserving, thus $w(a'_{n+1})=D_n$, as it would for the case where $\ell_n = 0$. Specifically, $f(a'_{n+1}) = s - a'_{n+1} - D_n$, as before.
We repeat the same correcting procedure for every succeeding job $j_n$ for which $\ell_n < 0$, if such job exists.
Notice that since $f$ is monotonically non-increasing, if $f(t_n) < 0$ for some job $j_n$ then $f(a_{n+1}) = f(a'_{n+1}) < 0$. Therefore, the algorithm is guaranteed to fail, as in Lemma \ref{thm:Committed_LowerBound_SingleServer_Natural}.
\end{proofof}



We now generalize our impossibility result for $1 \le C \le 3$ servers.
\\

\begin{proofof}{Theorem \ref{thm:Committed_LowerBound}}
Assume towards contradiction that there exists a committed algorithm $\Alg^C$ for $C$ servers with a bounded competitive ratio for $s < \ffrac{4}{C}$. We can construct a single server algorithm $\Alg^1$ for $s<4$ with bounded competitive ratio, contradicting Theorem \ref{thm:Committed_LowerBound_SingleServer}.
To do so, we translate every time unit for the $C$ server algorithm $\Alg^C$ to $C$ consecutive time slots for $\Alg^1$.
\end{proofof}


\section{Truthful Committed Scheduling}
\label{sec:Appendix_Committed_Truthful}

\newcommand{\tD}{\tilde{D}}

In Section \ref{sec:Committed_Truthful_PublicArrivalTimes} we described a committed scheduler that is truthful with respect to values, deadlines, and demands, but not necessarily with respect to arrival time.  In this section we show how to extend our construction to be fully truthful with respect to all parameters.  For ease of readability we will drop the parametrization with respect to $\omega$, and simply set $\omega = \frac{1}{2}$ in all invocations of earlier results.

Recall that our method for building responsive schedulers is to split each job's execution window into a simulation phase and an execution phase.  As discussed in Section \ref{sec:Committed_Truthful_PublicArrivalTimes}, the reason that the scheduler from Section \ref{sec:Committed_Truthful_PublicArrivalTimes} is not truthful with respect to arrival time is that a job may benefit by influencing the time interval in which the simulation phase is executed.  By declaring a later arrival, a job may shift the simulation to a later, less-congested time, increasing the likelihood that the simulator accepts the job.

Our strategy for addressing this issue is to impose additional structure on the timing of simulations.  Roughly speaking, we will imagining partitioning (part of) each job's execution window into many sub-intervals.  A simulation will be run for \emph{each} subinterval, and the job will be admitted if any of these simulations are successful.  Our method for selecting these \emph{simulation intervals} will be monotone: reporting a smaller execution window or a larger job can only result in smaller simulation intervals.  Using the truthful scheduling algorithm from Section \ref{sec:NonCommitted_Truthfulness} as a simulator will then result in an overall truthful scheduler.  The competitive ratio analysis will follow by extending our dual-fitting technique to allow multiple simulations for a single job.


\paragraph{Defining Simulation Intervals}

Our method of choosing sub-intervals will be as follows.  Choose a parameter $\sigma > 1$ to be fixed later; $\sigma$ will determine a minimal slackness constraint for our simulations.  Given slackness parameter $s$ and a job $\type_j = \langle v_j, D_j, a_j, d_j \rangle$, let $k_j$ be the minimal integer such that $2^{k_j} \geq 2 \sigma D_j$.  The value $2^{k_j}$ will be the minimal length of a simulation window.  Simulation intervals will have lengths that are powers of $2$, and endpoints aligned to consecutive powers of $2$.

We say an interval $[a,b]$ is \emph{aligned} for job $j$ if:
\begin{enumerate}
\item $[a,b] \subseteq [a_j, d_j]$,
\item $[b, b + (b-a)] \subseteq [a_j, d_j]$, and
\item $a = t \cdot 2^k$ and $b = (t+1) \cdot 2^k$ for some integers $t \geq 0$ and $k \geq k_j$.
\end{enumerate}
Write ${\cal C}_j$ for the collection of maximal aligned intervals for job $j$, where maximality is with respect to set inclusion.  For example, if $k_j = 2$ and $[a_j, d_j] = [9, 50]$, then ${\cal C}_j = \{ [12, 16], [16, 32], [32, 40], [40, 44] \}$.  Note that $[16, 20]$ is not in ${\cal C}_j$ because it is not maximal: it is contained in $[16, 32]$.  Also, $[32, 48]$ is not in ${\cal C}_j$ because it is not aligned for job $j$: the second condition of alignment is not satisfied, since $[48, 64] \not\subseteq [9, 50]$.

We refer to ${\cal C}_j$ as the simulation intervals for job $j$; it is precisely the set of intervals on which the execution of job $j$ will be simulated.  We now make a few observations about simulation intervals.

\begin{proposition}
\label{prop.aligned.nonempty}
If $\sigma \leq s/12$ then ${\cal C}_j$ is non-empty.
\end{proposition}
\begin{proof}
We prove the contrapositive.  If ${\cal C}_j$ is empty, then there is no subinterval of $[a_j, d_j]$ of the form $[t \cdot 2^{k_j}, (t+2) \cdot 2^{k_j}]$.  It must therefore be that $[a_j, d_j]$ is contained in an interval of the form $( t \cdot 2^{k_j}, (t+3) \cdot 2^{k_j} )$.  Thus $(d_j - a_j) < 3 \cdot 2^{k_j}$.  From the definition of $k_j$, we have $2^{k_j} < 4 \sigma D_j$, and hence $(d_j - a_j) < 12 \sigma D_j$.  Since job $j$ has slackness $s$, we conclude $\sigma > s / 12$.
\end{proof}

\begin{proposition}
\label{prop.aligned.large}
If ${\cal C}_j$ is non-empty then ${\cal C}_j$ is a disjoint partition of an interval $I \subseteq [a_j, d_j]$, with $|I| \geq \frac{1}{4}(d_j - a_j)$.
\end{proposition}
\begin{proof}
Disjointness follows because the intervals in ${\cal C}_j$ are aligned to powers of $2$ and are maximal.  That their union forms an interval follows from the fact that, for each $k$, the aligned intervals of length $2^k$ together form a contiguous interval.  It remains to bound the length of the interval $I$.

Choose $k$ such that the maximal-length interval in ${\cal C}_j$ has length $2^k$.  Choose $t_1$ and $t_2$ so that $a_j \in ( (t_1 - 1) 2^k, t_1 2^k]$ and $d_j \in [ t_2 2^k, (t_2 + 1) 2^k)$.  Then $(d_j - a_j) \leq (t_2 - t_1 + 2) \cdot 2^k$.  Also, since ${\cal C}_j$ contains an interval of length $2^k$, we must have $(t_2 - t_1) \geq 2$.  Moreover, each interval $[ t 2^k, (t+1)2^k]$ with $t_1 \leq t \leq t_2 - 1$ is aligned for job $j$, and hence $|I| \geq (t_2 - t_1 - 1) 2^k$.  We conclude
$$|I| \geq (t_2 - t_1 - 1) 2^k \geq (d_j - a_j) \cdot \frac{t_2 - t_1 - 1}{t_2 - t_1 + 2} \geq (d_j - a_j) \cdot \frac{1}{4}$$
where in the last inequality we used $(t_2 - t_1) \geq 2$.
\end{proof}

\paragraph{The Scheduling Mechanism: Single Server}

We now describe our truthful committed scheduler, denoted $\Alg_{TC}$.  We begin by describing the construction for the single-server case.  The main idea is a straightforward extension of the simulation methodology described in Section \ref{sec:Committed}.
%
%
%
For each job $j$ that arrives with declared type $\type_j = \langle v_j, D_j, a_j, d_j \rangle$, and for each subinterval $[a^{(i)}_j, b^{(i)}_j] \in {\cal C}_j$,
we will create a new \emph{phantom job} $\type^{(i)}_j = \langle v_j, 2 D_j, a^{(i)}_j, b^{(i)}_j \rangle$.  We will then employ the online, truthful, non-committed scheduling algorithm $\Alg_{\cal T}$ from Section \ref{sec:NonCommitted_Truthfulness}, using these phantom jobs as input.  If a phantom job $\type^{(i)}_j$ completes, then all subsequent phantoms for the corresponding job $j$ are removed from the input, and job $j$ is subsequently processed on the ``real'' server (using an EDF scheduler).  That is, a job is admitted if any of its phantom jobs complete; otherwise, if none of its phantoms completes, then it is rejected.  Note that since the phantom jobs have disjoint execution windows, it is known whether a given phantom completes before any subsequent phantom jobs arrive, and hence phantom jobs can be ``removed'' in an online fashion.


\begin{theorem}
\label{thm.truthful.reduction}
Choose $\sigma > 1$ and suppose $s \geq 12 \sigma$.  Then the scheduler $\Alg_{TC}$ described above is $2 \sigma$-responsive, truthful, and has competitive ratio bounded by
\[ \emph{\CompetitiveRatio}_{\Alg_{TC}}(s) \le 8 \cdot \emph{\CompetitiveRatio}_{\Alg}\left(\sigma \right). \]
\end{theorem}

We prove each property of the theorem in turn.  To establish responsiveness, note that the scheduler will always commit to executing a job $j$ by the end of the last interval in ${\cal C}_j$.  Since each aligned interval has length at most $2^{k_j}$, and since an aligned interval of length $\ell$ must end before time $d_j - \ell$ (condition 2 in the definition of aligned intervals), this endpoint occurs at least $2^{k_j} \geq 2 \sigma D_j$ time units before $d_j$.  This implies that the scheduler is $2 \sigma$-responsive.

We next bound the competitive ratio of the modified scheduler.
\begin{claim}
The competitive ratio of the scheduler $\Alg_{TC}$ described above is at most
\[ \emph{\CompetitiveRatio}_{\Alg_{TC}}(s) \le 8 \cdot \emph{\CompetitiveRatio}_{\Alg_{T}}\left(\sigma \right). \]
\end{claim}
\begin{proof}
Consider an input instance $\type$ with slackness $s$.
Let $\virtual{\type}$ denote the following ``phantom'' input instance: for each job $j$ in $\type$ we include all phantom jobs up to and including the first phantom accepted by $\Alg$, but not those that follow.  Note then that running $\Alg_{T}$ on inputs $\virtual{\type}$ generates the same total value as running $\Alg_{TC}$ on input instance $\type$.  Also note that the slackness of the phantom input instance is at least $\sigma$.

We prove the claim by constructing a feasible dual solution $(\alpha,\beta)$ satisfying \eqref{eq:DualConstraintSingleServer} and bounding its total cost.
Let $(\alpha^*,\beta^*)$ denote the optimal fractional solution of the dual program corresponding to $\virtual{\type}$.
We assume $\Alg$ induces an upper bound on the integrality gap for slackness $\sigma$. Therefore, the dual cost of $(\alpha^*,\beta^*)$ is at most $\CompetitiveRatio_{\Alg}(\sigma) \cdot v(\Alg_T(\virtual{\type})) = \CompetitiveRatio_{\Alg}(\sigma) \cdot v(\Alg_{TC}(\type))$.

The claim follows by applying the resizing lemma and the stretching lemma to $(\alpha^*,\beta^*)$.
First, we apply the resizing lemma for $f=2$, as each phantom corresponding to job $j$ has demand at most $2 D_j$. This increases the dual cost by a multiplicative factor of $2$.
Second, we apply the stretching lemma to all of the phantom jobs corresponding to job $j$, so that their execution windows remain disjoint and contiguous, their last deadline becomes $d_j$, and their earliest arrival time becomes $a_j$.  By Proposition \ref{prop.aligned.large}, this involves invoking the stretching lemma with $f = 4$.
Denote by $(\alpha', \beta')$ the resulting resized and stretched dual solution.
Finally, for each job $j$ we take $\alpha_j$ to be the maximum of the entries of $\alpha'$ corresponding to phantoms of $j$, and we take $\beta = \beta'$.

After applying both lemmas, we obtain a feasible dual solution that satisfies the dual constraints \eqref{eq:DualConstraintSingleServer}. The dual cost of the solution is at most:
\begin{equation*}
8 \cdot \emph{\CompetitiveRatio}_{\Alg_T}\left(\sigma\right) \cdot v(\Alg_{CT}(\type))
\end{equation*}
and therefore by applying the dual fitting theorem (Theorem \ref{thm:DualFitting}) we obtain our desired result.
\end{proof}

Finally, we argue that the resulting mechanism is truthful.

\begin{claim}
Scheduler $\Alg_{TC}$ is truthful, with respect to job parameters $\langle v_j, D_j, a_j, d_j \rangle$.
\end{claim}
\begin{proof}
Consider a job $j$ and fix the reports of other jobs.  Consider two types for job $j$, say $\type_j$ and $\type'_j$, with $\type_j$ dominating $\type'_j$.
Let ${\cal C}_j$ and ${\cal C}'_j$ denote the sets of simulation intervals under reports $\type_j$ and $\type'_j$, respectively.  We claim that for every interval $I' \in {\cal C}'_j$ there exists some $I \in {\cal C}_j$ such that $I' \subseteq I$.

Before proving the claim, let us show how it implies $\Alg_{TC}$ is truthful.  Recall from the definition of $\Alg_T$ that a job is successfully scheduled in the simulator if the set of times in which a higher-priority job is being run satisfies a certain downward-closed condition.  Moreover, the times in which higher-priority jobs are run is independent of the reported properties of lower-priority jobs, including all phantoms of job $j$.  Thus, a job $j$ is accepted if and only if there is some $I \in {\cal C}_j$ for which the corresponding phantom would complete in the simulator, and this is independent of the other intervals in ${\cal C}_j$.   (Note that this independence is the only point in the argument where we use the specific properties of algorithm $\Alg_T$, beyond truthfulness.)  But now reporting $\type'_j$ dominated by $\type_j$ can only result in smaller simulation intervals (in the sense of set inclusion), which can only result in lower acceptance chance for any given simulation interval by the truthfulness of $\Alg_T$.  Thus, if job $j$ is not accepted under type $\type_j$, it would also not be accepted under type $\type'_j$.

It remains to prove the claim about ${\cal C}'_j$ and ${\cal C}_j$.
%
It suffices to consider changes to each parameter of job $j$ separately.
Changing the value $v_j$ has no impact on the simulation intervals.  Increasing the demand $D_j$ can only raise $k_j$, which can only serve to exclude some intervals from being aligned.  
Likewise, increasing $a_j$ or decreasing $d_j$ can also only exclude some intervals from being aligned.  But if the set of aligned intervals is reduced, and some interval $[a,b]$ lies in ${\cal C}'_t$ but not in ${\cal C}_t$, then it must be that $[a,b]$ is an aligned interval under reports $\type_j$ and $\type'_j$, but is not maximal under report $\type_j$.  In other words, there must be some $[a',b'] \in {\cal C}_t$ such that $[a,b] \subseteq [a',b']$, as required.
%
%
%
\end{proof}

\paragraph{Extending to Multiple Servers}

We can extend our construction to multiple identical servers in precisely the same manner as in Theorem \ref{thm:Committed_MultipleServer_Migration_CompetitiveRatio}.  Specifically, when generating phantom jobs, we increase their demand by an additional factor of $11.656$.  As in Theorem \ref{thm:Committed_MultipleServer_Migration_CompetitiveRatio}, this allows us to argue that the simulated migratory schedule implies the existence of a non-migratory schedule of shorter phantom jobs, which in turn implies that passing accepted jobs to a global EDF scheduler results in a feasible schedule.  We obtain the following result.

\begin{theorem}
\label{thm.truthful.reduction.multiple}
Choose $\sigma > 1$ and suppose $s \geq (12 \cdot 11.656) \cdot \sigma$.  Then the scheduler $\Alg_{TC}$ described above is $2 \sigma$-responsive, truthful, and has competitive ratio bounded by
\[ \emph{\CompetitiveRatio}_{\Alg_{TC}}(s) \le (8 \cdot 11.656) \cdot \emph{\CompetitiveRatio}_{\Alg}\left(\sigma \right). \]
\end{theorem}



\section{Obtaining Theorem Statements from Section 1.1}
\label{sec:Appendix_Theorem_Statements}

The body of the paper describes the general results we obtain for truthful committed scheduling. In this appendix, we state the specific results we obtain by invoking these reductions on specific schedulers.  Specifically, the non-truthful scheduler \cite{LMNY13} and the truthful scheduler developed in Section \ref{sec:NonCommitted_Truthfulness}.  In each case, constant bounds can be obtained by plugging the algorithms directly. However, we can improve the constants by via a more careful analysis, using the dual-fitting analysis from the original algorithms.

For the non-truthful scheduling algorithm $\Alg$ from \cite{LMNY13}, the competitive ratio is bounded by explicitly constructing a feasible dual solution $(\alpha,\beta)$ and bounding its dual cost. The following bounds were obtained:
\begin{eqnarray}
\sum_j D_j \alpha_j
  & \,\, = \,\, &
v(\Alg(\type)) \cdot \left[1 + \Theta\left(\frac{1}{\sqrt[3]{s}-1}\right)\right]
\\
\sum_{i=1}^C \intop_0^{\infty} \beta_i(t)dt
  & \,\, = \,\, &
v(\Alg(\type)) \cdot \left[ 1 + \Theta\left(\frac{1}{\sqrt[3]{s}-1}\right) + \Theta\left(\frac{1}{(\sqrt[3]{s}-1)^2}\right) \right]
\end{eqnarray}
The same bounds are obtained for the truthful algorithm $\Alg_T$ (from Section \ref{sec:NonCommitted_Truthfulness}), however the power in the last asymptotic bound is $3$ instead of $2$.

When analyzing the competitive ratio of reductions for committed scheduling (Theorem \ref{thm:Committed_SingleServer_CompetitiveRatio}, Corollary \ref{thm:Committed_MultipleServer_NonMigration_CompetitiveRatio}, Theorem \ref{thm:Committed_MultipleServer_Migration_CompetitiveRatio}, and Theorem \ref{thm:Committed_FullTruthfulness}), in each case we apply the resizing lemma and the stretching lemma on the dual solution $(\alpha,\beta)$. However, the constant blowup from application of these lemmas only affects the $\beta$ term.  Accounting for this leads to improved constants in the resulting competitive ratios.

For example, applying Corollary \ref{thm:Committed_MultipleServer_NonMigration_CompetitiveRatio} to the algorithm $\Alg$ from \cite{LMNY13} and setting $\omega = 1/2$, one obtains a factor $4$ blowup.  Applying this blowup only to the $\beta$ term in the dual-fitting analysis of $\Alg$, one obtains a final competitive ratio of
\begin{align*}
& \left[1 + \Theta\left(\frac{1}{\sqrt[3]{s}-1}\right)\right] + 4 \left[ 1 + \Theta\left(\frac{1}{\sqrt[3]{s}-1}\right) + \Theta\left(\frac{1}{(\sqrt[3]{s}-1)^2}\right) \right] \\
=\ & 5 + \Theta\left(\frac{1}{\sqrt[3]{s}-1}\right) + \Theta\left(\frac{1}{(\sqrt[3]{s}-1)^2}\right)
\end{align*}
yielding the result described in Section~\ref{sec:Introduction_OurResults} as Theorem~\ref{thm.alg.responsive}.

As another example, applying Theorem \ref{thm.truthful.reduction} to the algorithm $\Alg_T$ from Section \ref{sec:NonCommitted_Truthfulness} yields a factor $8$ blowup.  Applying this blowup only to the $\beta$ term in the dual-fitting analysis, one obtains a competitive ratio of
\begin{align*}
& \left[1 + \Theta\left(\frac{1}{\sqrt[3]{s}-1}\right)\right] + 8 \left[ 1 + \Theta\left(\frac{1}{\sqrt[3]{s}-1}\right) + \Theta\left(\frac{1}{(\sqrt[3]{s}-1)^3}\right) \right] \\
=\ & 9 + \Theta\left(\frac{1}{\sqrt[3]{s}-1}\right) + \Theta\left(\frac{1}{(\sqrt[3]{s}-1)^3}\right).
\end{align*}
The extension to multiple servers follows the same approach as Theorem \ref{thm:Committed_MultipleServer_Migration_CompetitiveRatio}, requiring an additional application of the resizing lemma with a factor $f \approx 11.656$.  Again applying this only to the $\beta$ term, this increases the constant portion of the competitive ratio to $1 + 8 \cdot 11.656 \approx 94.248$, and increases the slackness requirement by an additional factor of $11.656$.  This
yields the result described in Section~\ref{sec:Introduction_OurResults} as Theorem~\ref{thm.general}, and restated as Theorem \ref{thm:Committed_FullTruthfulness}.





\end{document}